%% file: ms.tex
\documentclass[conference]{IEEEtran}
\IEEEoverridecommandlockouts
\usepackage{cite}
\usepackage{amsmath,amssymb,amsfonts}
\usepackage{algorithmic}
\usepackage{graphicx}
\usepackage{textcomp}
\usepackage{xcolor}
\def\BibTeX{{\rm B\kern-.05em{\sc i\kern-.025em b}\kern-.08em
    T\kern-.1667em\lower.7ex\hbox{E}\kern-.125emX}}

\usepackage{hyperref}

\usepackage{multirow}
\usepackage{booktabs}
\usepackage{todonotes}

\usepackage{extarrows}
\usepackage{clrscode3e}
\usepackage[ruled,vlined]{algorithm2e}

\usepackage{algorithmic}
\usepackage{multicol}

\usepackage{amsmath}
\usepackage{amssymb}
\usepackage{amsfonts}
\usepackage{amsthm}

\newcommand{\varfont}[1]{{{\scriptstyle\mathsf{#1}}}}
\newcommand{\var}[1]{{{\varfont{#1}}}}
\newcommand{\scr}{{\varfont{SCR}}}
\newcommand{\shr}{{\varfont{SHR}}}
\newcommand{\jkd}{{\varfont{JKD}}}
\newcommand{\mgp}{{\varfont{MGP}}}
\newcommand{\wox}{{\varfont{WOX5}}}
\newcommand{\clex}{{\varfont{CLEX}}}
\newcommand{\plt}{{\varfont{PLT}}}
\newcommand{\arf}{{\varfont{ARF}}}
\newcommand{\auxiaa}{{\varfont{AUXIAA}}}
\newcommand{\auxin}{{\varfont{AUXIN}}}

\usepackage{xargs}
\usepackage{todonotes}
\newcommandx{\todos}[2][1=]{\todo[inline,caption={},linecolor=red,backgroundcolor=red!25,bordercolor=red,#1]{\textbf{TODO: }#2}}
\newcommandx{\commentGeorgios}[2][1=]{\todo[inline,caption={},linecolor=red,backgroundcolor=red!25,bordercolor=red,#1]{\small#2\textbf{\\-- Georgios}}\xspace}
\newcommandx{\commentAlberto}[2][1=]{\todo[inline,caption={},linecolor=cyan,backgroundcolor=cyan!25,bordercolor=cyan,#1]{\small#2\textbf{\\-- Alberto}}\xspace}
\newcommandx{\commentAndrea}[2][1=]{\todo[inline,caption={},linecolor=purple,backgroundcolor=purple!25,bordercolor=purple,#1]{\small#2\textbf{\\-- Andrea}}\xspace}

\newcommand{\BB}{\mathbb{B}}

\newcommand{\RE}{\mathbb{R}}
\newcommand{\hX}{\hat{X}}
\newcommand{\hs}{\hat{s}}

\newcommand{\op}{\oplus}

\newcommand{\calH}{\mathcal{H}}
\newcommand{\calX}{\mathcal{X}}
\newcommand{\calO}{\mathcal{O}}

\newcommand{\MM}{\mathbb{M}}
\newcommand{\dt}{\tau}
\newcommand{\disc}{discrete-time}

\newtheorem{theorem}{Theorem}
\newtheorem{remark}{Remark}
\newtheorem{corollary}{Corollary}
\newtheorem{definition}{Definition}
\newtheorem{examplenum}{Example}

\begin{document}

\title{Minimization of Dynamical Systems over Monoids}

\author{
	\IEEEauthorblockN{
		 Georgios Argyris\IEEEauthorrefmark{1},
		 Alberto Lluch Lafuente\IEEEauthorrefmark{1},
		 Alexander Leguizamon Robayo\IEEEauthorrefmark{2},
		 Mirco Tribastone\IEEEauthorrefmark{3},\\
		 Max Tschaikowski\IEEEauthorrefmark{2}, and
		 Andrea Vandin\IEEEauthorrefmark{4},\IEEEauthorrefmark{1}
	 \IEEEauthorblockA{\IEEEauthorrefmark{1}DTU Compute, Technical University of Denmark, Denmark}
	 \IEEEauthorblockA{\IEEEauthorrefmark{2}Department of Computer Science, Aalborg University, Denmark}
	 \IEEEauthorblockA{\IEEEauthorrefmark{3}SysMA Research Unit, IMT Lucca, Italy}
	 \IEEEauthorblockA{\IEEEauthorrefmark{3}Institute of Economics and EMbeDS, Sant'Anna School for Advanced Studies, Pisa, Italy}
}
}

\maketitle

\begin{abstract}
Quantitative notions of bisimulation are well-known tools for the minimization of dynamical models such as Markov chains and ordinary differential equations (ODEs). In \emph{forward bisimulations}, each state in the quotient model represents an equivalence class and the dynamical evolution gives the overall sum of its members in the original model. Here we introduce generalized forward bisimulation (GFB) for dynamical systems over commutative monoids and develop a partition refinement algorithm to compute the coarsest one. When the monoid is $(\mathbb{R}, +)$, we recover 
probabilistic bisimulation for Markov chains and more recent forward bisimulations for 
nonlinear ODEs. 
Using $(\mathbb{R}, \cdot)$ we get 
nonlinear reductions for discrete-time dynamical systems and ODEs 
where each variable in the quotient model represents the product of original variables in the equivalence class. When the domain is a finite set such as the Booleans $\mathbb{B}$, we can apply GFB to Boolean networks (BN), a widely used dynamical model in computational biology. Using a prototype implementation of our minimization algorithm for GFB, we find disjunction- and conjunction-preserving reductions on 60 BN from two well-known  repositories, and demonstrate the obtained analysis speed-ups. We also provide the biological interpretation of the reduction obtained for two selected BN, and we show how GFB enables the analysis of a large one that could not be analyzed otherwise.
Using a randomized version of our algorithm we find product-preserving (therefore non-linear) reductions on 21 dynamical weighted networks from the literature that could not be handled by the exact algorithm.
\end{abstract}


\input{introduction}

\input{prelim}

\input{gfb}

\input{computeGFB}

\input{continuousTime}

\input{applications}

\input{conlcusions}

\bibliography{ms}


\end{document}

%% file: introduction.tex
	\section{Introduction}\label{sec:intro}

Bisimulation is a fundamental tool in computer science for abstraction and minimization, relating models by useful logical and dynamical properties~\cite{sangiorgi:book}. Originally developed for concurrent processes in a non-quantitative setting~\cite{DBLP:conf/tcs/Park81}, it has been extended to quantitative models 
such as, e.g., the 
notion of probabilistic bisimulation~\cite{Larsen19911} closely related to ordinary lumpability for Markov chains~\cite{BuchholzOrdinaryExact}.
\emph{Forward} bisimulations  relate states based on criteria that depend on their \emph{outgoing} transitions (as opposed to \emph{backward} bisimulations that depend on \emph{incoming} ones, e.g.,~\cite{De-Nicola:1990aa}).
When applied to a dynamical system (DS), forward bisimulations preserve 
sums of values of state variables. E.g., probabilistic bisimulation yields a quotient Markov process where each state represents an equivalence class preserving the sum of the probabilities of its members; forward bisimulation for reaction networks identifies equivalence classes among the chemical species that preserve the total concentration~\cite{concur15,Bioinf21};
forward differential equivalence (FDE) for nonlinear ordinary differential equations (ODEs) relates 
 variables preserving sums of their solutions~\cite{DBLP:conf/popl/CardelliTTV16}.

An attractive feature of bisimulation is that one can compute the largest bisimulation 
 using partition refinement, based on the pioneering solution for concurrent processes~\cite{paige1987three}.
Partition refinement algorithms start from an \emph{initial partition} of variables which is iteratively refined (i.e., its blocks get split) until the  partition is a bisimulation. Notably,  initial partitions are useful 
to tune the reduction. For example, one can \emph{separate} groups of variables according to given criteria 
to prevent that variables from different groups will be aggregated together.
 %
 This makes bisimulation an effective approach for the minimization of complex
 DS, adding to 
 cross-disciplinary methods originated in e.g., chemical and performance engineering~\cite{okino1998,DBLP:conf/splc/Tribastone14}, control theory~\cite{antoulas}, and systems biology~\cite{Snowden:2017aa}.

Thus far, one can identify two common properties of existing forward bisimulations for DS: 
they preserve sums of state values, and the DS variables
take
real~$\mathbb{R}$ values.
There are, however, motivations to
generalize this setting. A forward bisimulation for ODEs can be seen as a special case of  \emph{linear lumping}~\cite{okino1998}, a minimization achieved by a 
linear projection of the state space 
by a matrix that encodes the 
partition of  state variables. However, one may be also interested in \emph{nonlinear lumpings} where each state in the reduced model represents a nonlinear transformation of original variables~\cite{LI1994343}.

Another motivating question tackled in this paper is the generalization of the domain on which the DS evolves. Forward bisimulation is not currently applicable to DS  that evolve over finite domains. Consider, e.g., the DS
\begin{equation}\label{eq:simple}
    \begin{split}
		x_1(k+1) & =  x_2(k) \vee x_3(k) \\
		x_2(k+1) & = x_1(k) \vee x_3(k)\\
		x_3(k+1) & = \neg x_3(k) \wedge (x_1(k) \vee x_2(k))
	\end{split}
\end{equation}
where variables $x_1$, $x_2$, and $x_3$ are defined over the Booleans $\BB=\{0,1\}$, and  $k$ denotes discrete time.
This is 
a Boolean network (BN), an established model of biological systems~\cite{KAUFFMAN1969437}.

Here we develop a more abstract notion of forward bisimulation, \emph{generalized forward bisimulation} (GFB), for a DS over a (commutative) monoid. We show that this is a conservative extension with respect to the literature because we recover available notions of forward bisimulation for DS when the monoid is $(\mathbb{R}, +)$. However, it is more general. For example, over the monoid $(\BB,\vee)$ one can prove that variables $x_1$ and $x_2$ in~\eqref{eq:simple} are \emph{GFB equivalent}, i.e., we can rewrite the model in terms of $x_1 \vee x_2$ and $x_3$. Indeed, by computing the disjunction of the left- and right-hand-side of $x_1$ and $x_2$ in~\eqref{eq:simple} we get
\begin{align*}
	x_1(k\!+\!1) \!\vee\! x_2(k\!+\!1) & =    x_2(k) \vee x_3(k) \vee x_1(k) \vee x_3(k)
	\\
 &=   x_3(k) \vee (x_1(k) \vee x_2(k)).
\end{align*}
By using the derived variable $x_{1,2}\equiv x_{1}\vee x_2$, we get the \emph{GFB-reduced model}
\begin{equation}\label{eq:simpleReduced}
	\begin{split}
		x_{1,2}(k+1) & = x_3(k) \vee x_{1,2}(k) \\
		x_3(k+1) & = \neg x_3(k) \wedge  x_{1,2}(k).
	\end{split}
\end{equation}

This can be used in place of the original model if one is not interested in the individual values of $x_1$ and $x_2$, but only in their disjunction.

Here  we show that GFB satisfies  desirable properties for bisimulations. 
\begin{enumerate}
	\item Over any commutative monoid $(\mathbb{M}, \op)$, GFB characterizes $\op$-preserving reductions, in the sense that any DS with fewer state variables which coincides with $\op$-operations of original state variables must necessarily be the quotient of a GFB. This generalizes characterization results for Markov chains~\cite{Larsen19911}, chemical reaction networks~\cite{PNAScttv}, and nonlinear ODEs~\cite{DBLP:conf/popl/CardelliTTV16}.
	Notably, our characterization result also covers the asymptotic dynamics, often of interest when analyzing DS (see, e.g.,~\cite{hopfensitz2013attractors}). We show that GFB preserves all \emph{attractors}, i.e., the 
	states towards which the DS tends to evolve and remain.
	
	\item GFB can be computed by a partition refinement algorithm. We develop a \emph{template} algorithm which hinges on the computation of a formula whose decidability and complexity depend on the domain and the right-hand sides of the DS under study. In general, this can be undecidable. However, when the monoid is $(\mathbb{R}, +)$ our algorithm reduces to that for FDE for nonlinear ODEs~\cite{DBLP:conf/popl/CardelliTTV16,PNAScttv,DBLP:conf/qest/CardelliTTV18}. Instead, when the domain is $\mathbb{B}$, the problem corresponds to Boolean satisfiability.
	
	\item For polynomial ODEs and  monoid $(\RE,\cdot)$, we obtain, to the best of our knowledge,
	the first algorithm for nonlinear reduction in (randomized) polynomial time.
	
	\item GFB is effective in practice, both in terms of reduction power and of obtained analysis speed-ups.
\end{enumerate}

Previous results are essentially agnostic to whether the time evolution of the DS is continuous or discrete. More specifically, the criteria for probabilistic bisimulation~\cite{Larsen19911} are the same for both continuous-time and discrete-time Markov chains. Similarly, FDE equivalently applies to both a nonlinear ODE system in the form $\partial_t x = f(x)$ (where $\partial_t$ denotes time derivative) and to a discrete-time nonlinear DS in the form $x(k+1) = f(x(k))$. With GFB, instead, more care has to be taken because this verbatim correspondence does not hold any longer. For this reason, we first develop GFB for discrete-time DS. Then, we consider continuous time by studying GFB for DS over the reals relating to, and extending, results for ODEs.

\emph{Applications.} Using a prototype implementation in the tool ERODE~\cite{cardelli2017erode}, 
we apply GFB  to case studies from different domains.  We consider Boolean and multi-valued 
networks~\cite{KAUFFMAN1969437,thomas1995dynamical}, where the latter allows for finer degrees of activation than just 0/1 as in~\eqref{eq:simple}. These models are known to suffer from state-space explosion, making model reduction appealing (see, e.g.,~\cite{argyris2021reducing}). 
Two selected case studies from the literature showcase the physical intelligibility of GFB reductions, and a third one shows how GFB can enable the analysis of otherwise intractable    BNs. In the three case studies, we show how  initial partitions can be devised using domain knowledge. 
For example, we show how  ($\BB$,$\land$) allows to identify and abstract away from distinct \emph{sub-models} (biological pathways);  we show how finite monoids and operations $\min$ and $\max$ allow studying  \emph{full model (de)activation}, meaning that we obtain reductions that 
track groups of components whose activation denote the (de)activation of different mechanisms of the model.
We also perform a large-scale validation of GFB on 60 Boolean and multi-valued 
networks from established repositories (GinSim~\cite{NALDI2009134}, BioModelsDB~\cite{BioModels2020}). We show that \emph{default} initial partitions can be synthesized automatically. We also show
that GFB is \emph{useful} due to its high reduction power and analysis speed-ups. 
%
We also consider real-valued DS. We study a case study of a higher-order Lotka-Volterra model~\cite{PredatorPreyHO}, and we perform a large-scale validation on 72 weighted networks from the Netzschleuder repository~\cite{Netzschleuder}. 



\section{Related work}\label{sec:related}

Most of the literature about model minimization can be found for DS over the reals. In this context, the 
framework of exact lumping considers 
reductions by means of both linear and nonlinear operators~\cite{Li19891413,tomlin1997effect}. The aforementioned notions of bisimulation for Markov chains and FDE, including its stochastic variants~\cite{DBLP:journals/bioinformatics/CardelliPTTVW21,DBLP:conf/birthday/CardelliTTV17}, can be seen as specific linear reductions that are induced by a partition of the state space. Indeed, this corresponds to a specific type of minimization known as \emph{proper lumping}, where each original variable is represented by only one variable in the reduced model~\cite{okino1998}. Since also GFB is developed in this style, it too can be seen as a special case of exact lumping. However, the coarsest GFB can be computed in randomized polynomial time when the dynamics is described by polynomials over the monoids $(\RE,+)$ or $(\RE,\cdot)$, see Section~\ref{sec:comp}. Instead, the computation of exact lumpings hinges, in the case of polynomial dynamics, on symbolic computations with worst-case exponential  complexity~\cite[Section 2.2]{LI1994343}.

Relying on polynomial invariants~\cite{DBLP:conf/tacas/GhorbalP14,DBLP:conf/atva/BartocciKS19}, $\mathcal{L}$-bisimulation~\cite{DBLP:journals/lmcs/Boreale19,DBLP:journals/scp/Boreale20} can be seen as a generalization of backward differential equivalence (BDE)~\cite{DBLP:journals/tcs/TschaikowskiT14a,DBLP:journals/tcs/TschaikowskiT14,DBLP:journals/tcs/CardelliTTV19a}, a backward-type bisimulation for non-linear ODEs, and is thus complementary to FDE (hence, GFB), as discussed in~\cite{DBLP:journals/lmcs/Boreale19,DBLP:journals/scp/Boreale20,BBLTTV21Lics}. It is also worth noting that neither BDE nor $\mathcal{L}$-bisimulation allow for model reduction through nonlinear transformations, in contrast to GFB. Similarly to $\mathcal{L}$-bisimulation, consistent abstraction (aka bisimulation)~\cite{DBLP:journals/tac/PappasLS00,DBLP:journals/tac/PappasS02,bisimulation_lin_sys_Schaft} is complementary to GFB. Indeed, for a so-called observation function, the coarsest consistent abstraction gives rise to a minimal reduced DS which coincides with the original one up to the chosen observation function. Instead, computing the coarsest GFB corresponds to 
finding an observation function which induces a minimal consistent abstraction. Hence, GFB reduces across observation functions, while consistent abstraction reduces with respect to a given observation function. Moreover, in contrast to consistent abstraction, GFB considers the subclass of observation functions induced by equivalence relations. To the best of our knowledge, the computation of an  observation function yielding a minimal reduced model has been investigated for linear dynamics only~\cite{DBLP:journals/tac/PappasLS00}.

In the model checking community there has been a large amount of work on predicate abstraction techniques, including counterexample-guided abstraction refinement (CEGAR)~\cite{DBLP:conf/cav/ClarkeGJLV00} and abstraction invariant checking (see, e.g., the recent work in~\cite{DBLP:conf/cav/MoverCGIT21} for polynomial ODEs). 
Our approach consists in a model-to-model abstraction reduction, where the final result is a reduced model specified in the same specification language as the original one (i.e., systems of equations). The abstraction function is discovered automatically, through a partition refinement algorithm that, at each iteration, computes and checks an invariant of a specific form that corresponds to the GFB conditions (roughly, that the current abstraction candidate is a sound one). This could be seen as a sort of counterexample-guided abstraction refinement process. The class of abstraction functions considered, however, is not arbitrary: only monoidal operations are considered, which yields two main advantages: efficiency in the partition refinement algorithm, and meaning for the class of application domains considered (e.g. full/partial activation in biological models).

Reduction techniques exist  for BNs. Boolean backward equivalence (BBE) is a backward-type bisimulation~\cite{argyris2021reducing}, in line exact lumpability for Markov chains~\cite{BuchholzOrdinaryExact} and BDE. Hence,  it can be shown that BBE and GFB (applied to BNs) are not comparable.
%
Other approaches for BN reduction are based on
\emph{variable absorption} (e.g., \cite{naldi2011dynamically,veliz2011reduction}) where selected variables are  \emph{absorbed} by the update functions of their target variables 
by replacing all occurrences of the absorbed variables with their update functions. These approaches are complementary to GFB because they do not compute exact reductions.

%
	%
	%
	%

%% file: prelim.tex
\section{Preliminaries}\label{sec:prel}

In this section we formalize the notion of \disc{} DS and of  attractor for \disc\ DS~\cite{AttractorsDefs}, 
and notation considered in this paper. Then, we provide a running example used throughout the text.


\begin{definition}[Dynamical System]\label{def:bn}
	A discrete-time DS is a pair $D = (X, F)$ where $X\!=\!\{x_1,\ldots,x_n\}$ are variables and $F\!=\!\{f_{x_1},\ldots,f_{x_n}\}$ is a set of update functions, where $f_{x_i} : \MM^X \rightarrow \MM$ is the update function of variable $x_i$. Elements of $\MM^X\!$ are \emph{states}.
	The solution (simulation) of $D$ for initial state $s(0) \!\in\! \MM^X$ is given by the sequence $(s(k))_{k\geq0}$, where $s(k\!+\!1) \!=\! F(s(k))$ for all $k \!\geq \!0$. 
\end{definition}



%
We use $R$ to denote an equivalence relation over $X$, and $\calX_R$ the induced partition. We often do not distinguish among an equivalence relation and its induced partition. If not mentioned, we assume that $\op : \MM \times \MM \to \MM$ is such that $(\MM,\op)$ is a commutative monoid with neutral element $0_\op$. Moreover, $G^I$ denotes the set of all (total) functions from $I$ to $G$ and $f[a / b]$ is the term arising by replacing each occurrence of $a$ by $b$ in $f$.

As running example
we use a BN from~\cite{azpeitia2010single} that describes cell differentiation.
Deeper biological interpretation and its reduction will be given in Section~\ref{sec:applications}.

\begin{examplenum}\label{ex:1}
	Let $(X, F)$ be a discrete-time DS with Boolean variables
	$X = \{\scr, \shr, \jkd, \mgp, \wox, \clex, \plt, \arf,$ $\auxiaa, \auxin \}$
	and 
	 update function
	$F : \BB^X \!\to\! \BB^X$ with
	\begin{align*}
		%
		f_\scr & = \shr \land \scr \land (\jkd \lor \neg \mgp) \ \  &  f_\clex & = \shr \land \clex \\
		f_\shr & = \shr \qquad & f_\plt & = \arf \\
		f_\jkd & = \shr \land \scr  & f_\arf & = \neg \auxiaa \\
		f_\mgp & = \shr \land \scr \land \neg \wox & f_\auxiaa & = \neg \auxin \\
		f_\wox & = \arf \land \shr \land \scr \land \neg \clex & f_\auxin & = \auxin
	\end{align*}
	Monoids for the DS are $(\BB,\op)$, $\op \!\in\! \{\land, \!\lor, \!\mathit{XOR} \}$, with neutral elements $1,0,0$. 
\end{examplenum}



\begin{definition}[Attractor]\label{def:attractor}
	Let $D\!=\!(X,F)$ be a discrete-time DS. A non-empty set $A \subseteq \MM^X$ is called attractor of $D$ (wrt  some given topology of $\MM^X$) whenever
	\begin{itemize}
		\item $A$ is invariant under $F$, that is, $F(A) \subseteq A$;
		\item there is an open neighborhood $B$ of $A$ s.t.  for any $v \in B$ there exists a $\nu \geq 1$ such that $F^n(v) \in A$ for all $n \geq \nu$. $B$ is called a \emph{basin of attraction} of $A$.
	\end{itemize}
\end{definition}


\begin{examplenum}\label{ex:transition}
	Let $s = (0,0,0,0,0,1,1,1,0,1) \in \BB^{X}$ denote a state of the DS from Example~\ref{ex:1} where only the variables $\clex,\plt,\arf,\auxin$ are active. 
	By applying the update functions we get  $F(s)=s'=(0,0,0,0,0,0,1,1,0,1) \in \BB^X$, where $\plt$, $\arf$ and $\auxin$ are active. 
	If we apply the update functions again, the system remains in the same state, i.e., $F(s')=s'$, meaning that $\{ s' \}$ is an attractor. 
\end{examplenum}


%% file: gfb.tex
\section{Generalized Forward Bisimulation}\label{sec:GFB}

Here we define generalized forward bisimulation (GFB), the notion of GFB reduction, and
show that GFB reductions preserve the original model dynamics.


\begin{definition}[Generalized Forward Bisimulation]\label{defBE}
	Let $D=(X,F)$ be a discrete-time DS, $(\MM,\op)$ a commutative monoid and $\calX_R$ a partition of $X$. Then, $\calX_R$ is a GFB when the following formula holds true:
	\begin{multline*}
		\forall s,s' \in \MM^X .
		\operatorname*{\bigwedge} \limits_{C \in \calX_R } \Bigl( \bigoplus\limits_{x_i \in C } s_{x_i} \!=\! \bigoplus \limits_{x_i \in C } s'_{x_i} \Bigr)
		\\
		\ \Longrightarrow 
		\operatorname*{\bigwedge} \limits_{C \in \calX_R } \Bigl(  \bigoplus\limits_{x_i \in C }f_{x_i}(s)  \!=\! \bigoplus \limits_{x_i \in C } f_{x_i}(s') \Bigr) .
	\end{multline*}
	The homomorphism of $R$, denoted by $\psi_R : \MM^X \to \MM^{\calX_R}$, is given by
	\[
	\psi_R(s)_C = \bigoplus\limits_{x_i \in C} s_{x_i} , \quad \text{for all} \ C \in \calX_R .
	\]
\end{definition}

\begin{examplenum}\label{ex:2}
For $\op \!=\! \land$,	$\calX_R \! =\! \{ C, \{\plt\}, \{\arf\}, \{\auxiaa\},$ $ \{\auxin\} \}$ is a GFB for our running example, where $C \!=\! \{ \scr, \shr, \jkd, \mgp, \wox, \clex \}$. 
	This means that the 
	running example can be rewritten solely in terms of conjunctions over all variables in C, and the other individual variables. 
	To this end, we first note that for all $x_i \notin C$ we have that $f_{x_i}$ is independent of any $x_j \in C$.~\footnote{However, the original system is not trivially decoupled in variables in $C$ and variables not in $C$, because $\arf$ appears in the update function of $\wox$.}
	Moreover, 
	the update functions of $\wox$ and $\clex$ contain terms $\neg\clex$ and $\clex$, respectively, therefore the conjunction of their update functions (and of all variables in $C$) can be simply rewritten as $0$ since: 
	$\bigwedge\limits_{x_i \in C} f_{x_i}(s) = s_{\clex} \land \neg s_{\clex} \land ( \ldots ) = 0$.
\end{examplenum}






\begin{definition}[Reduced DS]\label{def:red}
	The reduction $D/R$ of a discrete-time DS $D = (X,F)$ for an equivalence  $R$,  is the DS $(\calX_R, F_R)$ with $F_R = (f_C)_{C \in \calX_R}$ such that
	\begin{align*}
		f_C = \operatorname*{\bigoplus} \limits_{x_i \in C} f_{x_i} [x_k / 0_\oplus : x_k \notin \hX] [x_{i_{C'}} / x_{C'} : C' \in \calX_R ] ,
	\end{align*}
	where $x_{i_C} \in C$ is a representative of $C \in \calX_R$ and $\hX = \{ x_{i_C} : C \in \calX_R \}$ is the set of all representatives.
\end{definition}



\begin{examplenum}\label{ex:4}
	We compute the reduced DS of our running example for the GFB $\calX_R$ from Example~\ref{ex:2}. We choose $\jkd$ as representative of $C$, while the  representative for the other (singleton) blocks is obvious.
	With this, we obtain
	\begin{align*}
		f_C & = \bigwedge\limits_{x_k \in C} f_{x_k} [x_k / 1 : x_k \notin \hX] [x_{i_{C'}} / x_{C'} : C' \in \calX_R ] \\
		& = 1 \land 1 \land \big(C \lor \neg 1\big) \land 1 \land 1 \land 1 \land 1
		\\ & \qquad
		\land 1 \land \neg 1 \land \{\arf\} \land 1 \land 1 \land \neg 1 \land 1 \land 1 = 0
	\end{align*}
	For all other blocks, instead, we obtain
	\begin{align*}
		f_{\{\plt\}} & \!=\! {\{\arf\}},   &
		f_{\{\arf\}} & \!=\! \neg {\{\auxiaa\}},  
		\\
		f_{\{\auxiaa\}} & \!=\! \neg {\{\auxin\}}, &
		f_{\{\auxin\}} & \!=\! {\{\auxin\}}
	\end{align*}
\end{examplenum}

\begin{remark}
	We note that, syntactically, the reduced DS depends on the choice of representatives. However, if $R$ is a GFB, then Theorem~\ref{thm:quant} guarantees that such choice does not affect the \emph{semantics} of the reduced DS. 
\end{remark}

We now show that $D$ and $D/R$ have \emph{same dynamics up to $\psi_R$} iff $R$ is a GFB. 


\begin{theorem}[GFB characterization via model dynamics]\label{thm:quant}
	Fix a DS $D=(X,F)$, a partition $\calX_R$ of $X$, $D/R = (\calX_R,F_R)$, and a commutative monoid $(\MM, \op)$.  
	Then, 
	$R$ is a GFB iff for any initial state $s_0 \in \MM^X$  the solutions of $D$ and $D/R$ for $s_0$ and $\hs_0 = \psi_R(s_0)$, respectively, are equal up to $\psi_R$. That is:
	\[
	\hs_k = \psi_R(s_k), \quad \text{for~} k \geq 0,
	\]
	where $s_{k+1} = F(s_k)$ and $\hs_{k+1} = F_R(\hs_k)$.
\end{theorem}

\begin{proof}[Proof of Theorem~\ref{thm:quant}]
	Let $R$ be a GFB, pick $s_0 \in \MM^X$ and set $\hs_0 = \psi_R(s_0) \in \MM^{\calX_R}$. We next show that $\hs_k = \psi_R(s_k)$ by induction over $k \geq 0$. Since the base case $k = 0$ is true by construction, we can turn to the induction step. For $k \geq 0$, we obtain
	\[
	\hs_{k+1} = F_R(\hs_k) = F_R(\psi_R(s_k)) = \psi_R(F(s_k)) = \psi_R(s_{k+1}) ,
	\]
	where the second identity follows from the induction hypothesis, while the third identity follows from the definition of $F_R$ and the fact that $R$ is a GFB. Conversely, if $\hs_k = \psi_R(s_k)$ for all $k \geq 0$, we can conclude for $k = 0$ and arbitrary $s_0 \in \MM^X$ that
	\[
	\psi_R(F(s_0)) = \psi_R(s_1) = \hs_1 = F_R(\hs_0) = F_R(\psi_R(s_0)) ,
	\]
	thus showing that $R$ is a GFB.
\end{proof}

Theorem~\ref{thm:quant} readily implies the following result on attractors.

\begin{corollary}\label{cor:attractor}
	Let $D=(X,F)$ be a DS, $(\MM, \op)$ a commutative monoid, $R$ a GFB and $D/R = (\calX_R,F_R)$. 
	Then, we have the following two (equivalent) statements.
	\begin{itemize}
		\item If $A \subseteq \MM^X$ is an attractor of $D$, then $\psi_R(A) \subseteq \MM^{\calX_R}$ is an attractor of $D/R$.
		\item If $A \subseteq \MM^{\calX_R}$ is not an attractor of $D/R$, then 
		$\psi_R^{-1}(A) \subseteq \MM^{X}$ is not an attractor of $D$.
	\end{itemize}
\end{corollary}


\begin{examplenum}\label{example:reducedAttractors}
	We consider the attractor $s' = \{ (0,0,0,0,0,0,$ $1,1,0,1) \}$ from Example~\ref{ex:transition}. The homomorphism $\psi_R$ maps the attractor to $\psi_R(s')=\big\{ (0,1,1,0,1) \big\}$. Corollary~\ref{cor:attractor} ensures that the set $\psi_R(s')$ is an attractor of the reduced system $D/R$. Indeed, by applying the update functions $F_R$ to  $(0,1,1,0,1)$, the reduced system remains at the same state, and thus $\psi_R(s')$ is invariant under $F_R$.
\end{examplenum}


%% file: computeGFB.tex
\section{Computation of the coarsest GFB}\label{sec:comp}
Computing the coarsest GFB that refines a given initial partition is based on the classic partition refinement algorithm~\cite{paige1987three} where the blocks of an initial partition are iteratively refined (or split) until a GFB is obtained. The coarsest GFB is obtained when the initial partition contains one block only. Different  initial partitions can be useful 
to tune  reductions to preserve variables of interest (see, e.g., Section~\ref{sec:applications}).
%
%
Here we prove that there exists a unique coarsest GFB that refines a given initial partition, and that the algorithm computes it.

\begin{theorem}\label{thm:maximal}
	Let $D=(X,F)$ be a discrete-time DS, and $\calX_R$ a partition of $X$. 
	There exists a unique coarsest GFB $\calH$ that refines $\calX_R$.
\end{theorem}

\begin{proof}[Proof of Theorem~\ref{thm:maximal}]
	Fix arbitrary GFBs $\sim_1, \ldots, \sim_\nu \subseteq R$ and let $\calH_1, \ldots,$ $\calH_\nu$ be the corresponding partitions, i.e., $\calH_i = X_{\sim_i}$. Moreover, let $\sim_\ast := \big( \bigcup_{i = 1}^m \sim_i \big)^*$ and $\calH^\ast := X_{\sim_\ast}$, where the asterisk denotes transitive closure of a relation. At last, let $x_{i_{H^\ast}} \in H^\ast$ denote some representative of $H^\ast \in \calH^\ast$. With this, pick an arbitrary $H^\ast \in \calH^\ast$. By construction of $\calH^\ast$, there exist $x_0, \ldots,x_k \in X$ and $i_0,\ldots,i_{k-1} \in \{1,\ldots,\nu\}$ so that $\{x_0,\ldots,x_k\} = H^\ast$, $x_k = x_{i_{H^\ast}}$ and
	$x_j \sim_{i_j} x_{j + 1}$ for all $0 \leq j \leq k - 1$. Moreover, for any $G^\ast \in \calH^\ast$ and $1 \leq i \leq \nu$, there exist (unique) $G^i_1,\ldots,G^i_{m_i} \in \calH_i$ such that $\biguplus_{l = 1}^{m_i} G^i_l = G^\ast$. Since $x_j \sim_{i_j} x_{j + 1}$ and $\calH_{i_j}$ is a GFB, we obtain
	\begin{align*}
		\operatorname*{\bigoplus}_{x_\iota \in G^\ast} f_{x_\iota} & = \operatorname*{\bigoplus}_{l = 1}^{m_{i_j}} \operatorname*{\bigoplus}_{x_\iota \in G^{i_j}_l} f_{x_\iota} \\
		& = \operatorname*{\bigoplus}_{l = 1}^{m_{i_j}} \operatorname*{\bigoplus}_{x_\iota \in G^{i_j}_l} f_{x_\iota}[x_j / 0_\op][x_{j+1} / (x_j \op x_{j+1})] \\
		& = \operatorname*{\bigoplus}_{x_\iota \in G^\ast} f_{x_\iota}[x_j / 0_\op][x_{j+1} / (x_j \op x_{j+1})]
	\end{align*}
	Since $\{x_0,x_1,\ldots,x_k\} = H^\ast$ and $x_k = x_{i_{H^\ast}}$, an application of the argument for all $0 \leq j \leq k - 1$ implies that $\operatorname*{\bigoplus}_{x_\iota \in G^\ast} f_{x_\iota}$ is equivalent to
	\begin{align*}
		\operatorname*{\bigoplus} \limits_{x_\iota \in G^\ast} f_{x_\iota} [x_k / 0_\oplus : x_k \in H^\ast, x_k \neq x_{i_{H^\ast}}] [x_{i_{H^\ast}} / \bigoplus \limits_{x_l \in H^\ast} x_l ]
	\end{align*}
	Since the choice of $G^\ast, H^\ast \in \calH^\ast$ was arbitrary, we infer that $\calH^\ast$ is a GFB.
\end{proof}

A partition refinement algorithm for computing GFB needs a condition to tell: (i) if the current partition is a GFB, and, if not, (ii) how to split its blocks towards getting a GFB. Definition~\ref{defBE} can only be used for Point (i).
Theorem~\ref{thm:bin:fe} below provides a binary, relation-driven, characterization of GFB which allows for Point (ii). 
The intuition is that, by applying such binary characterization pairwise to all variables in each block of the current partition, we get the sub-blocks in which they should be split in the next iteration.

\begin{theorem}[Binary Characterization of GFB]\label{thm:bin:fe}
	Let $D=(X,F)$ be a DS, $(\MM,\op)$ a commutative monoid, $R$ an equivalence relation on $X$, and $\calX_R$ the induced partition. Then, $\calX_R$ is a GFB if and only if for any $(x_i,x_j) \in R$ with $x_i \neq x_j$, the following formula holds (where $0_\op$ is the neutral element of $\op$):
	\[
	\Psi^{\calX_R}_{x_i,x_j} \equiv \operatorname*{\bigwedge}\limits_{C \in \calX_R } \Bigl( \bigoplus\limits_{x_k \in C }f_{x_k} = \bigoplus\limits_{x_k \in C } f_{x_k} [x_i / 0_\oplus] [x_j / (x_i \oplus x_j)]	\Bigr)
	\]
\end{theorem}

\begin{proof}[Proof of Theorem~\ref{thm:bin:fe}]
	Let us assume first that $\calX_R$ is a GFB, pick an arbitrary $(x_i,x_j) \in R$ and pick the unique $C' \in \calX_R$ such that $x_i,x_j \in C'$. With this, define $s' := s[x_i \mapsto 0_\op][x_j \mapsto s_{x_i} \op s_{x_j}]$ for an arbitrary $s \in \MM^X$, where $s[x_k \mapsto b]_{x_k} = b$ and $s[x_k \mapsto b]_{x_l} = s_{x_l}$ for all $b \in \MM$ and $x_l \neq x_k$. Then, since $\op$ is commutative and associative and because  $\calX_R$ is a GFB, we have that
	\begin{align}\label{eq:thm:bin:fe:2}
		\operatorname*{\bigwedge} \limits_{C \in \calX_R } \Bigl( \bigoplus\limits_{x_i \in C }f_{x_i}(s)  = \bigoplus \limits_{x_i \in C } f_{x_i}(s') \Bigr).
	\end{align}
	Since the choice of $(x_i,x_j) \in R$ and $s \in \MM^X$ was arbitrary, we infer that $\Psi^{\calX_R}_{x_i,x_j}$ is valid. For the converse, let us assume that $\Psi^{\calX_R}_{x_i,x_j}$ holds true for all $(x_i,x_j) \in R$ and pick any two $s,s' \in \MM^X$ such that
	\begin{align}\label{eq:thm:bin:fe}
		\operatorname*{\bigwedge} \limits_{C \in \calX_R } \Bigl( \bigoplus\limits_{x_i \in C } s_{x_i} = \bigoplus \limits_{x_i \in C } s'_{x_i} \Bigr)
	\end{align}
	With this, pick for any $C \in \calX_R$ some arbitrary representative $x_{i_C} \in C$ and let $\hX = \{ x_{i_C} : C \in \calX_R \}$ be the set of all representatives. For any $(x_i,x_j) \in R$, define $s_{i \to j} := s[x_i \mapsto 0_\op, x_j \mapsto s_{x_i} \op s_{x_j}]$. With this, the fact that $\op$ is commutative and associative ensures the existence of a sequence $x_{i_1},x_{i_2},...,x_{i_k}$ for which $\hs = (((s_{i_1 \to i_2})_{ i_2 \to i_3}) \ldots)_{i_{k-1} \to i_k}$ is such that
	\begin{align*}
		\operatorname*{\bigwedge} \limits_{C \in \calX_R } \Bigl( \bigoplus\limits_{x_i \in C } s_{x_i} = \bigoplus \limits_{x_i \in C } \hs_{x_i} \Bigr) ,
	\end{align*}
	$\hs_{x_i} = 0_\op$ for all $x_i \notin \hX$ and $\hs_{x_{i_C}} = \bigoplus\limits_{x_i \in C } s_{x_i}$ for all $C \in \calX_R$. Since $\Psi^{\calX_R}_{x_{i_l},x_{i_{l+1}}}$ is valid for all $1 \leq l \leq k - 1$, we obtain
	\[
	\operatorname*{\bigwedge} \limits_{C \in \calX_R } \Bigl( \bigoplus\limits_{x_i \in C }f_{x_i}(s)  = \bigoplus \limits_{x_i \in C } f_{x_i}(\hs) \Bigr).
	\]
	A similar argument for $s'$ ensures that there is an $\hs'$ such that $\hs'_{x_i} = 0_\op$ for all $x_i \notin \hX$, $\hs'_{x_{i_C}} = \bigoplus\limits_{x_i \in C } s'_{x_i}$ for all $C \in \calX_R$ and
	\begin{align*}
		\operatorname*{\bigwedge} \limits_{C \in \calX_R } \Bigl( \bigoplus\limits_{x_i \in C } s'_{x_i} & = \bigoplus \limits_{x_i \in C } \hs'_{x_i} \Bigr) , \\
		\operatorname*{\bigwedge} \limits_{C \in \calX_R } \Bigl( \bigoplus\limits_{x_i \in C }f_{x_i}(s') & = \bigoplus \limits_{x_i \in C } f_{x_i}(\hs') \Bigr) .
	\end{align*}
	Thanks to~(\ref{eq:thm:bin:fe}), we infer that $\hs = \hs'$. This, in turn, implies the desired relation~(\ref{eq:thm:bin:fe:2}), thus showing that $\calX_R$ is a  GFB if and only if $\Psi^{\calX_R}_{x_i,x_j}$ is valid for all $(x_i,x_j) \in R$.
\end{proof}

The binary characterization tells us that we can rewrite an $\oplus$-expression of the update functions of a block of a GFB  in terms of $\oplus$-expressions of pairs of GFB equivalent variables $x_i$ and $x_j$. This can be done by successively moving, pair by pair, all variables of a GFB equivalence class to a chosen  representative.




\begin{examplenum}\label{ex:3}
	Let us consider the GFB $\calX_R$ from Example~\ref{ex:2}, the only non-singleton block $C \in \calX_R$, and  the variables $\shr,\jkd \in C$. With $\op = \land$ and $0_\land = 1$, we obtain
	\begin{align*}
		\bigwedge\limits_{x_k \in C} f_{x_k}  & = \shr \land \scr \land (\jkd \lor \neg \mgp)
		\\ & \qquad
		\land \shr \land \shr \land \scr \land \shr
		\land \scr \land \neg \wox \land \arf \land  \\
		& \qquad \shr\land \scr \land \neg \clex \land \shr \land \clex \\
		& = 0 \\
		& = 1 \land \scr \land ((\jkd \land \shr) \lor \neg \mgp)
		\\ & \qquad
		\land 1 \land 1 \land \scr \land 1
		\land \scr \land \neg \wox \land \arf  \\
		& \qquad \land 1 \land \scr \land \neg \clex \land 1 \land \clex
		\\ &
		= \bigwedge_{x_k \in C} f_{x_k}[\shr/1,\jkd/(\shr \land \jkd)]
	\end{align*}
	For any other block 
	the clause is trivially true because $\shr$ and $\jkd$ appear only in the update functions of variables in $C$. Hence, $\Psi^{\calX_R}_{\shr, \jkd}$ is valid.
	Similarly, we can show that $\Psi^{\calX_R}_{x_i, x_j}$ is valid for all $(x_i,x_j) \!\in\! R, x_i \!\neq\! x_j$. Hence $\calX_R$ is a GFB.
\end{examplenum}

\begin{algorithm}[tp!]
		\begin{algorithmic}[1]
			\WHILE{\textbf{true}}
			\STATE $\calH^\prime \leftarrow \emptyset$
			\FORALL{$H \in \calH$}
			\STATE $R \leftarrow \{(x_i,x_j) \in H \times H : \text{if } x_i \neq x_j, $  then $\Psi^\calH_{x_i,x_j} \text{ and } \Psi^\calH_{x_j,x_i} \}$
			
			\STATE $\calH' \leftarrow \calH' \cup (H/R)$
			\ENDFOR
			
			\IF{$\calH = \calH'$}
			\STATE \textbf{return} $\calH$
			\ELSE
			\STATE $\calH \leftarrow \calH'$
			\ENDIF
			\ENDWHILE
		\end{algorithmic}
	\caption{Compute the coarsest GFB that refines an initial partition $\calH$ for DS $(X,F)$.}
	\label{algorithm}
\end{algorithm}




The next result addresses the algorithmic computation of the coarsest GFB.

\begin{theorem}\label{thm:computesmaximal}
	Let $D=(X,F)$ be a discrete-time DS and $X_R$ a partition. Algorithm~\ref{algorithm} computes the coarsest GFB refining $R$ by deciding at most $\calO(|X|^3)$ instances of formula $\Psi^\calH_{x_i,x_j}$. If $\MM$ is finite, any formula $\Psi^\calH_{x_i,x_j}$ is decidable.
\end{theorem}

\begin{proof}[Proof of Theorem~\ref{thm:computesmaximal}]
	Pick the coarsest GFB $\calH_\ast$ that refines $X_R$ using Theorem~\ref{thm:maximal}. With this, set $\calH_0 := X_R$ and define for all $k \geq 0$ and $H \in \calH_k$
	\begin{align*}
		R_k(H) & := \{(x_i,x_j) \in H \times H : x_i \neq x_j \Rightarrow \Psi^{\calH_k}_{x_i,x_j} \land \Psi^{\calH_k}_{x_j,x_i} \}  \\
		\calH_{k+1} & := \bigcup_{H \in \calH_k} H / R^\ast_k(H) ,
	\end{align*}
	where $R^\ast_k(H)$ denotes the transitive closure of $R_k(H)$. By construction, $R_k(H)$ is reflexive and symmetric, thus implying $\operatorname*{\bigoplus} \limits_{x_i \in H} f_{x_i}(s) = \operatorname*{\bigoplus} \limits_{x_i \in H} f_{x_i}(\tilde{s})$ for all $s \in \MM^X$, $H \in \calH_k$, where
	\begin{align*}
		\tilde{s} = s[x_j \mapsto 0_\oplus : x_j \notin \hX_{k+1}][x_{i_{C'}} \mapsto \bigoplus_{x_j \in C'} s_{x_j} : C' \in \calH_{k+1}]
	\end{align*}
	and $x_{i_C} \in C$ is a representative of class $C \in \calH_{k+1}$, while $\hX_{k+1} = \{ x_{i_C} : C \in \calH_{k+1} \}$. (Note that $H \in \calH_k$, while $C \in \calH_{k+1}$ and $\hX_{k+1}$ is defined using $\calH_{k+1}$.) This implies that $R_k$ is transitive. Indeed, for any $(x_i,x_j),(x_j,x_k) \in R_k$ and $s' \in \MM^X$, the previous equation ensures for state $s := s'[x_i \mapsto 0_\op, x_k \mapsto s'_{x_i} \op s'_{x_k}]$ and any $H \in \calH_k$ that
	\begin{align*}
		\operatorname*{\bigoplus} \limits_{x_i \in H} f_{x_i}(s) = \operatorname*{\bigoplus} \limits_{x_i \in H} f_{x_i}(\tilde{s}') = \operatorname*{\bigoplus} \limits_{x_i \in H} f_{x_i}(\tilde{s}'')
		= \operatorname*{\bigoplus} \limits_{x_i \in H} f_{x_i}(\tilde{s}''') ,
	\end{align*}
	where
	\begin{align*}
		\tilde{s}' & \! = \! s[x_l \! \mapsto \! 0_\oplus : x_l \notin \hX_{k+1}][x_{i_{C'}} \! \mapsto \! \bigoplus_{x_j \in C'} s_{x_j} : C' \in \calH_{k+1}] , \\
		\tilde{s}'' & \! = \! s'[x_l \! \mapsto \! 0_\oplus : x_l \notin \hX_{k+1}][x_{i_{C'}} \! \mapsto \! \bigoplus_{x_j \in C'} s'_{x_j} : C' \in \calH_{k+1}] , \\
		\tilde{s}''' & \! = \! s'[x_i \! \mapsto \! 0_\op, x_j \! \mapsto \! 0_\op, x_k \! \mapsto \! s'_{x_i} \op s'_{x_j} \op s'_{x_k}] .
	\end{align*}
	Hence, $R_k^\ast = R_k$ and the expression $H/R$ is indeed well-defined in Algorithm~\ref{algorithm}. Further, a proof by induction over $k \geq 1$ shows that a) $\calH_\ast$ is a refinement of $\calH_k$ and b) $\calH_k$ is a refinement of $\calH_{k-1}$. Since $\calH_\ast$ is a refinement of any $\calH_k$, it holds that $\calH_\ast = \calH_k$ if $\calH_k$ is a GFB partition. Since $X$ is finite, b) allows us to fix the smallest $k \geq 1$ with $\calH_k = \calH_{k-1}$. This, in turn, implies that $\calH_{k-1}$ is a GFB. To see the complexity statement, we note that the algorithm can perform at most $|X|$ refinements, while each iteration compares $\calO(|X|^2)$ pairs. For the decidability, instead, we first note that the finiteness of $\MM$ ensures the finiteness of $\op \subseteq \MM \times \MM$ and any $f_{x_i}  \subseteq \MM^X \times \MM$. Hence, checking
	\[
	\operatorname*{\bigwedge}\limits_{C \in \calH} \Bigl( \bigoplus\limits_{x_k \in C }f_{x_k} = \bigoplus\limits_{x_k \in C} f_{x_k} [x_i / 0_\oplus] [x_j / (x_i \oplus x_j)]	\Bigr)
	\]
	amounts to a finite number of checks over finite sets and is thus decidable.
\end{proof}

The decidability of $\Psi^\calH_{x_i,x_j}$ for $\MM$ infinite is less immediate. Indeed, since deciding $\Psi^\calH_{x_i,x_j}$ amounts to deciding identities between functions, decidability 
over infinite domains critically hinges on the nature of the update functions. For instance, if $\MM = \RE$, the conditions of $\Psi^\calH_{x_i,x_j}$ require one to decide the equivalence of real-valued functions. If $\op = +$ and update function terms arise through addition and multiplication of variables and may contain minima and maxima expressions, the problem is double exponential~\cite{DBLP:conf/popl/CardelliTTV16}. If also exponential and trigonometric functions are allowed, the problem becomes undecidable~\cite{RealsFragmentUndecidable}.

We thus study the complexity of deciding $\Psi^\calH_{x_i,x_j}$ when $(f_{x_i})_{x_i \in X}$ are polynomials and $\op \in \{+,\cdot\}$. In such a case, checking $\Psi^\calH_{x_i,x_j}$ amounts to deciding whether the polynomials
\[
\bigoplus\limits_{x_k \in C} f_{x_k} \quad \text{and} \quad \bigoplus\limits_{x_k \in C} f_{x_k} [x_i / 0_\oplus] [x_j / (x_i \oplus x_j)]
\]
are equal. In case of the real and complex field, this question is equivalent to polynomial identity testing for which no holistic algorithms with polynomial time complexity are known~\cite{DBLP:journals/eatcs/Saxena09}.\footnote{The common holistic approach rewrites a polynomial into a sum of monomials. Hence, if $\op = \cdot$ and all $f_{x_k}$ have, say, 2 monomials, a direct computation of the monomials of $\op_{x_k \in C} f_{x_k}$ requires $\calO(2^{|C|})$ steps.}
Fortunately, the following result readily follows from the Schwartz-Zippel lemma~\cite{DBLP:journals/eatcs/Saxena09}.

\begin{theorem}\label{thm:comlex}
	Let $D=(X,F)$ be a discrete-time DS and $\calX_R$ a partition. Then, if $(f_{x_i})_{x_i \in X}$ are polynomials over some (sufficiently large) field $\MM$ and $\op \in \{+,\cdot\}$, Algorithm~\ref{algorithm} runs in randomized polynomial time. More specifically, assume that $\Psi^\calH_{x_i,x_j}$ is false and that it involves polynomials of degree less or equal $d$. Then, for any finite set $S \subseteq \MM$, any $C \in \calH$ and a uniformly sampled $v \in S^X$, we have 
	\[
	\mathbb{P}\Bigl\{ \bigoplus\limits_{x_k \in C}f_{x_k}(v) = \bigoplus\limits_{x_k \in C} f_{x_k}[x_i / 0_\oplus] [x_j / (x_i \oplus x_j)](v) \Bigr\} \leq \frac{d}{|S|} ,
	\]
	where $\mathbb{P}\{A\}$ denotes the probability of event $A$. In particular, one obtains a polynomial time randomized algorithm whenever $\MM$ has more than $d$ elements.
\end{theorem}

%% file: continuousTime.tex
\section{Continuous-time DS}\label{sec:cont}


We relate GFB to continuous-time DS, showing how GFB encapsulates existing  bisimulations for (nonlinear)  ODEs. Thus, in what follows we consider DS with domain $\mathbb{R}$. 
We can study minimizations for an ODE system  $\partial_t v(t) = \Phi(v(t))$ (where $\partial_t$ denotes time derivative) using GFB on its time discretization $(X,F)$, where $F(s) = s + \dt \Phi(s)$. 
Standard results imply that the approximation error between the ODEs and its time discretization vanishes if $\dt$ approaches zero~\cite{Gear:1971}.


\subsection{Exact lumpability}
GFB-type reductions can be captured by exact lumpability, an  established reduction notion for ODEs~\cite{Li19891413,tomlin1997effect}. Indeed, exact lumping must not be necessarily induced by a partition of the variables. However, we will show that when an exact lumping on an ODE system is described by the homomorphism $\psi_R$ of an equivalence relation $R$ if and only if it is a GFB for its discretization, provided higher-order discretization errors are ignored.
We start with the definition of exact lumping~\cite{Li19891413}.





\begin{definition}\label{def:el}
	Given an ODE system $\partial_t v(t) = \Phi(v(t))$ with a differentiable function $\Phi : \RE^X \to \RE^X$, a twice differentiable function $\psi : \RE^X \to \RE^{\hat{X}}$ is an exact lumping if $|\hat{X}| < |X|$ and there is a unique differentiable function $\hat{\Phi} : \RE^{\hat{X}} \to \RE^{\hat{X}}$ such that for any $v : [0;T] \to \RE^X$ satisfying $\partial_t v(t) = \Phi(v(t))$, it holds that $\partial_t \psi(v(t)) = \hat{\Phi}(\psi(v(t)))$ for all $t \in [0;T]$.
\end{definition}


Consider, e.g., the model
\begin{align*}
	{\partial_t}v_{x_1} &= v_{x_1}\\
	{\partial_t}v_{x_2} & = v_{x_2}.
\end{align*}
Then, $\psi(v_{x_1},v_{x_2}) = v_{x_1} v_{x_2}$ is an exact lumping since
\begin{align*}
\partial_t \psi(v) & = (\partial_{x_1} \psi (v), \partial_{x_2} \psi (v) ) \cdot \Phi(v)
\\ & =
(v_{x_2}, v_{x_1}) \cdot (\partial_t v_{x_1}, \partial_t v_{x_2} )^T
\\ & =
 2 v_{x_1} v_{x_2} = 2 \psi(v)\,
\end{align*}
where superscript $T$ denotes the transpose of a vector.
We can observe that this can be discovered using GFB on the time discretization of the ODE system, given by
\[f_{x_1}(s)  = s_{x_1} + \dt s_{x_1} \ \text{ and }\ f_{x_2}(s)  = s_{x_2} + \dt s_{x_2}.\]
%
Indeed $\calX_R = \{ \{x_1,x_2\} \}$ is a GFB over $(\mathbb{R}, \cdot)$ since
\begin{align*}
	f_{x_1} \!\cdot\! f_{x_2} &=
	(x_1 \!+\! \dt x_1) \cdot (x_2 \!+\! \dt x_2)
		\\&=
	x_1 x_2 \!+\! 2 \dt x_1 x_2 \!+\! \dt^2 x_1 x_2
	\\&=
	(f_{x_1} \!\cdot\! f_{x_2})[x_2 / 1, x_1 / x_1 x_2].
\end{align*}
 This shows that $\psi_R$ is indeed an exact lumping. The next result formalizes this relationship.


\begin{theorem}\label{thm:cont}
	Given $\partial_t v(t) = \Phi(v(t))$ with a differentiable function $\Phi : \RE^X \to \RE^X$, consider the DS $D_{\dt} = (X,F)$ with $F(s) = s + \dt \Phi(s)$, where $\dt > 0$. Further, let us assume that $\op : \RE \times \RE \to \RE$ is twice differentiable and that $(\RE,\op)$ is a commutative monoid. Then, for any partition $\calX_R$ of $X$:
	\begin{enumerate}
		\item If $R$ is a GFB of all $D_{\dt}$, then $\psi_R$ is an exact lumpability of $\partial_t v(t) = \Phi(v(t))$.
		\item If $\psi_R$ is linear, then $R$ is a GFB of all $D_{\dt}$ if and only if $\psi_R$ is an exact lumping of $\partial_t v(t) = \Phi(v(t))$.
	\end{enumerate}
\end{theorem}

\begin{proof}[Proof of Theorem~\ref{thm:cont}]
	See proof of Theorem~\ref{thm:cont:char}.
\end{proof}


With the exception of the important case where $\psi_R$ is linear, Theorem~\ref{thm:cont} does not address whether GFB is also a necessary condition for exact lumpability. Indeed, it turns out that a characterization requires to relax formula $\Psi_{x_i,x_j}^{\calX_R}$ to, roughly speaking, ignore the terms of (higher) order $\dt^2, \dt^3, ...$ and so on.
This is exemplified next.

\begin{examplenum}
\label{example:log:counter}
Consider $\partial_t v(t) = \Phi(v(t))$ where $\Phi$ is given by
\begin{align*}
	\partial_t v_{x_1} & = v_{x_1} \log(v_{x_2}) & \text{and} & & 
    \partial_t v_{x_2} & = v_{x_2} \log(v_{x_1})
\end{align*}	
when both right-hand sides are defined, and zero otherwise. Then, for monoid $(\RE,\cdot)$ and $\calX_R = \{ \{x_1,x_2\} \}$, it holds that $\psi_R(v_{x_1},v_{x_2}) = v_{x_1} \cdot v_{x_2}$ is an exact lumping, while $\calX_R$ is not a GFB. In order to see this, we start by noting that
\begin{align*}
\partial_t \psi_R(v) & = v_{x_2} \partial_t v_{x_1} + v_{x_1} \partial_t v_{x_2} \\
& = v_{x_1} v_{x_2} (\log(v_{x_2}) + \log(v_{x_1})) \\
& = \psi_R(v) \log(\psi_R(v))
\end{align*}
At the same time, the ODE discretization of the model is
	\begin{align*}
		f_{x_1} & = x_1 + \dt x_1 \log(x_2) , &
		f_{x_2} & = x_2 + \dt x_2 \log(x_1) .
	\end{align*}
With this, we obtain
\begin{align*}
     f_{x_1} f_{x_2} & = (x_1 + \dt x_1 \log (x_2)) (x_2 + \dt x_2 \log (x_1)) \\
     & = \psi(v) + \tau \psi(v) \log(\psi(v)) + \tau^2 \psi(v) \log(x_1) \log(x_2)
\end{align*}
Since the higher-order term $\dt^2 x_1 x_2 \log(x_1) \log(x_2)$ cannot be expressed in terms of $x_1 x_2$, we conclude that $\calX_R$ is not a GFB.
\end{examplenum}

We now  characterize exact lumpings of the form $\psi_R$, accounting for Example~\ref{example:log:counter} and generalizing Theorem~\ref{thm:cont}. As anticipated, we 
 ignore higher-order terms $\calO(\dt^2)$ when checking $\Psi_{x_i,x_j}^{\calX_R}$, where $\calO$  is the \emph{big O} notation from numerical analysis.

\begin{theorem}\label{thm:cont:char} Given $\partial_t v(t) = \Phi(v(t))$ with a differentiable vector field $\Phi : \RE^X \to \RE^X$, consider the DS $D_{\dt} = (X,F)$ with $F(s) = s + \dt \Phi(s)$ where $\dt > 0$. Let us assume that $\op : \RE \times \RE \to \RE$ is twice differentiable and that $(\RE,\op)$ is a commutative monoid. Then, for any partition $\calX_R$ of $X$, function $\psi_R$ is an exact lumping iff for all $(x_i,x_j) \in R$ with $x_i \neq x_j$ formula $\Psi^{\calX_R}_{x_i,x_j}$ is valid up to $\calO(\dt^2)$, that is
\begin{equation}
\begin{split}
		\label{eq:psi:with:o}
		&\operatorname*{\bigwedge}\limits_{C \in \calX_R } \Bigl( \bigoplus\limits_{x_k \in C }f_{x_k} + \calO(\dt^2)  =
		\\ &\qquad
		\bigoplus\limits_{x_k \in C } f_{x_k} [x_i / 0_\oplus] [x_j / (x_i \oplus x_j)] + \calO(\dt^2) \Bigr) .
	\end{split}
\end{equation}

\end{theorem}

\begin{proof}[Proof of Theorem~\ref{thm:cont:char}]
	To improve readability, we write $\psi$ instead of $\psi_R$ in the present proof. Since $\op$ is twice differentiable by assumption, so is $\psi = (\psi_H)_{H \in \calX_R}$. For any $H \in \calX_R$, Taylor's theorem thus ensures
	\begin{align*}
		\psi_H(F(s)) &= \psi_H(s + \dt \Phi(s)) \\
		& = \psi_H(s) + (\partial_s \psi_H)(s + \dt \Phi(s)) \cdot \dt \Phi(s) + \calO(\dt^2) \\
		& = \psi_H(s) + \dt \cdot (\partial_s \psi_H)(s + \dt \Phi(s)) \cdot \Phi(s) + \calO(\dt^2)
	\end{align*}
	
	We begin by assuming that $\psi$ is an exact lumping. Then, with $\partial_t v(t) = \Phi(v(t))$, by~\cite{tomlin1997effect} the derivative of $t \mapsto \psi_H(v(t))$ can be written as a function of $\big(\psi_H(v(t))\big)_{H \in \calX_R}$. Since $v(0) \in \RE^X$ can be chosen arbitrarily, there is thus a function $\wp_H$ such that $\wp_H(\psi(s)) = (\partial_s \psi_H)(s + \dt \Phi(s)) \cdot \Phi(s) + \calO(\dt)$ for all $s \in \RE^X$. Indeed, by~\cite{tomlin1997effect}, there exists an $\hat{f}$ such that $\partial_t \psi(v(t)) = \hat{f}(\psi(v(t))$. At the same time, the chain rule yields $\partial_t \psi_H(v(t)) = (\partial_s \psi_H)(v(t)) \cdot \Phi(v(t))$. Setting $v(0) = s + \tau \Phi(s)$ for an arbitrary $s$, we thus get
    \begin{align*}
    \hat{f}_H(\psi(s + \dt \Phi(s))) & = (\partial_s \psi_H)(v(0)) \cdot \Phi(v(0)) \\
    & = (\partial_s \psi_H)(s + \tau \Phi(s)) \cdot \Phi(s + \tau \Phi(s)) \\
    & = (\partial_s \psi_H)(s + \tau \Phi(s)) \cdot \Phi(s) + O(\tau)
    \end{align*}
    With this, we can set $\wp_H = \hat{f}_H$ and conclude for all $s \in \RE^X$
	\begin{align*}
		\psi_H(F(s)) & = \psi_H(s) + \dt \cdot \wp_H(\psi(s)) + \calO(\dt^2)
	\end{align*}
	Since $H \in \calH$ can be chosen arbitrarily, following the argumentation from the proof of Theorem~\ref{thm:bin:fe}, we infer that for all $(x_i,x_j) \in R$ with $x_i \neq x_j$ formula $\Psi^{\calX_R}_{x_i,x_j}$ is valid up to $\calO(\dt^2)$. For the converse, let us assume that for all $(x_i,x_j) \in R$ with $x_i \neq x_j$ formula $\Psi^{\calX_R}_{x_i,x_j}$ is valid up to $\calO(\dt^2)$. Then, Taylor's theorem yields as before
	\begin{align*}
		\psi_H(F(s)) & = \psi_H(s) + \dt \cdot (\partial_s \psi_H)(s + \dt \Phi(s)) \cdot \Phi(s) + \calO(\dt^2)
	\end{align*}
	With this and the validity of the aforementioned $\Psi^{\calX_R}_{x_i,x_j}$, the argumentation from the proof of Theorem~\ref{thm:bin:fe} ensures the existence of functions $(\wp_H)_{H \in \calX_R}$ over $\RE^{\calX_R}$ such that
	\begin{align*}
		\psi_H(F(s)) & = \psi_H(s) + \dt \cdot \wp_H(\psi(s)) + \calO(\dt^2)
	\end{align*}
	for all $H \in \calX_R$ and $s \in \RE^X$. Hence, with $\partial_t v(t) = \Phi(v(t))$ and $v(0) = s + \dt \Phi(s)$, the chain rule implies
    \begin{align*}
    \partial_t \psi_H(v(0)) & = (\partial_s \psi_H)(v(0)) \cdot \Phi(v(0)) \\
    & = (\partial_s \psi_H)(s + \dt \Phi(s)) \cdot \Phi(s + \dt \Phi(s)) \\
    & = (\partial_s \psi_H)(s + \dt \Phi(s)) \cdot \Phi(s) + O(\dt) ,
    \end{align*}
thus yielding $\dt \cdot \partial_t \psi_H(v(0)) = \dt \cdot \wp_H(\psi(v(0))) + O(\dt^2)$. With this, we obtain that $\psi$ is an exact lumping. This completes the proof of Theorem~\ref{thm:cont:char}. We next turn to the proofs of 1) and 2) of  Theorem~\ref{thm:cont}. For 1), we note that $\Psi^{\calX_R}_{x_i,x_j}$ is valid up to $\calO(\dt^2)$ for all $(x_i,x_j) \in R$ when $R$ is a GFB. Instead, for 2) we observe that for a linear $\psi_R$ there are no higher-order terms, i.e., $\calO(\dt^2) = 0$. This two observations, combined with the foregoing discussion, yield statements 1) and 2).
\end{proof}

Theorem~\ref{thm:cont} is related to geometric integration where it has been shown~\cite[Section IV.1]{Hairer06} that discrete-time approximations preserve invariants of continuous-time DS only when these are linear or quadratic, but not if they are cubic or of higher degree. In contrast, Theorem~\ref{thm:cont:char} provides a one-to-one correspondence between continuous- and discrete-time invariants by dropping the higher order terms. Additionally, Theorem~\ref{thm:cont} and~\ref{thm:cont:char} allow in contrast to~\cite{Hairer06} for the algorithmic computation of (nonlinear) invariants.

We end the subsection by noting that if $(f_{x_i})_{x_i \in X}$ are polynomials, then \eqref{eq:psi:with:o} can be checked algorithmically by representing  polynomials as sums of monomials and by dropping afterwards all monomials containing a term $\tau^\nu$ with $\nu \geq 2$. 


\subsection{Forward differential equivalence and Markov chains}
Using the results of this section we can relate GFB with analogous bisimulations for DS. We start by restating the notion of forward differential equivalence (FDE) from~\cite{DBLP:conf/popl/CardelliTTV16}.

\begin{definition}[FDE]\label{def:FDE}
	Let us consider an ODE system $\partial_t v(t) = \Phi(v(t))$ with a differentiable function $\Phi : \RE^X \to \RE^X$. A partition $\calX_R$ of $X$ is an FDE if $\psi_R$ in case of $\op = +$ is an exact lumpability.
\end{definition}

The next result follows from Theorem~\ref{thm:cont}, relating GFB and  FDE~\cite{DBLP:conf/popl/CardelliTTV16}. 

\begin{corollary}\label{cor:fde}
	Given $\partial_t v(t) = \Phi(v(t))$ with a differentiable vector field $\Phi : \RE^X \to \RE^X$, and the DS $D_{\dt} = (X,F)$ with $F(s) = s + \dt \Phi(s)$, where $\dt > 0$. Then, for $(\RE,+)$, we have that $R$ is a GFB of all $D_{\dt}$ iff $R$ is an FDE of $\partial_t v(t) = \Phi(v(t))$.
\end{corollary}

\begin{proof}[Proof of Corollary~\ref{cor:fde}]
	Set $\op = +$ in Theorem~\ref{thm:cont}.
\end{proof}

Similarly, the next corollary relates GFB with
continuous-time Markov chains~\cite{BuchholzOrdinaryExact} and probabilistic bisimulation of discrete-time Markov chains~\cite{Larsen19911}.

\begin{corollary}\label{cor:prob:bis}
	Let $(X,Q)$ be a continuous-time Markov chain with states $X$ and transition rate matrix $Q \in \RE^{X \times X}$. Consider the DS $D_{\dt} = (X,F)$ with $F(s) = s + \dt Q^T s$ where $\dt > 0$. Then,
	$D_{\dt}$ is an embedded discrete-time Markov chain of 
	$(X,Q)$ for sufficiently small $\dt > 0$. With this, for monoid $(\RE,+)$ the following three conditions are equivalent: 1. $R$ is a GFB of all $D_{\dt}$; 2. $R$ is an ordinary lumpability of $(X,Q)$;  3. $R$ is a probabilistic bisimulation of all $D_{\dt}$ that describe a discrete-time Markov chain.
\end{corollary}

\begin{proof}[Proof of Corollary~\ref{cor:prob:bis}]
	The vector of transient probabilities of the Markov chain at time $t \geq 0$ satisfies the forward Kolmogorov equations $\partial_t \pi(t) = Q^T \pi(t)$. Moreover, by~\cite{DBLP:conf/popl/CardelliTTV16}, an equivalence relation $R$ over $X$ is an ordinary lumpability if and only if $R$ is an FDE the forward Kolmogorov equations. With this, Corollary~\ref{cor:fde} yields
	the equivalence of 1) and 2). The equivalence of 2) and 3), instead, is a well-known fact~\cite{BuchholzOrdinaryExact}.
\end{proof}

\begin{remark}
	The discussion shows that $\Psi^\calH_{x_i,x_j}$ in Algorithm~\ref{algorithm} can be decided in polynomial time for probabilistic bisimulation and FDE of polynomial differential equations~\cite{PNAScttv,DBLP:conf/popl/CardelliTTV16}, which can be in principle also extended to differential algebraic equations~\cite{DBLP:journals/nc/CardelliTT20}.
\end{remark}


\subsection{Attractors of continuous-time DS} The notion of attractor from Definition~\ref{def:attractor} also exists for continuous-time dynamics~\cite{AttractorsDiscretized}.

\begin{definition}[Attractor]
	Consider an ODE system $\partial_t v(t) = \Phi(v(t))$ with a differentiable vector field $\Phi : \RE^X \to \RE^X$. A compact nonempty set $A \subseteq \RE^X$ is an attractor (aka asymptotically stable) if there exists an open neighborhood $B$ of $A$ such that for any $\varepsilon > 0$ there is some time $t' \geq 0$ such that for any $v[0] \in B$, the solution of $\partial_t v(t) = \Phi(v(t))$ with $v(0) = v[0]$ satisfies $d(v(t),A) \leq \varepsilon$ for all $t \geq t'$. Here, $d(v(t),A) = \min_{a \in A} d(v(t),a)$ and distance $d$ is induced, similarly to $B$, by some norm.
\end{definition}


The next result from~\cite{AttractorsDiscretized} essentially ensures that attractors of an ODE system can be approximated by attractors of its discrete-time discretization.

\begin{theorem}[\cite{AttractorsDiscretized}]\label{thm:attractor:disc}
	Given $\partial_t v(t) = \Phi(v(t))$ with a differentiable vector field $\Phi : \RE^X \to \RE^X$, let $A \subseteq \RE^X$ be an attractor of $\partial_t v(t) = \Phi(v(t))$. Then, for any $\dt > 0$, there exists a set $A(\dt) \subseteq \RE^X$ such that
	\begin{itemize}
		\item $F(A(\dt)) \subseteq A(\dt)$, where $F(s) = s + \dt \Phi(s)$ and;
		\item The sets $A(\dt)$ converge to the set $A$ in the Hausdorff metric as $\dt \to 0$.
	\end{itemize}
\end{theorem}

Corollary~\ref{cor:attractor} and Theorem~\ref{thm:attractor:disc} allow to use GFB to argue on attractors of ODEs. Less importantly,  Theorem~\ref{thm:attractor:disc} does not explicitly provide basins of attraction for the sets $A(\dt)$. However, $A(\dt)$ are attractors when the discrete topology is used. 

%% file: applications.tex
\section{Applications}\label{sec:applications}
\begin{figure}[t]
	\centering
	\includegraphics[width=0.7\linewidth]{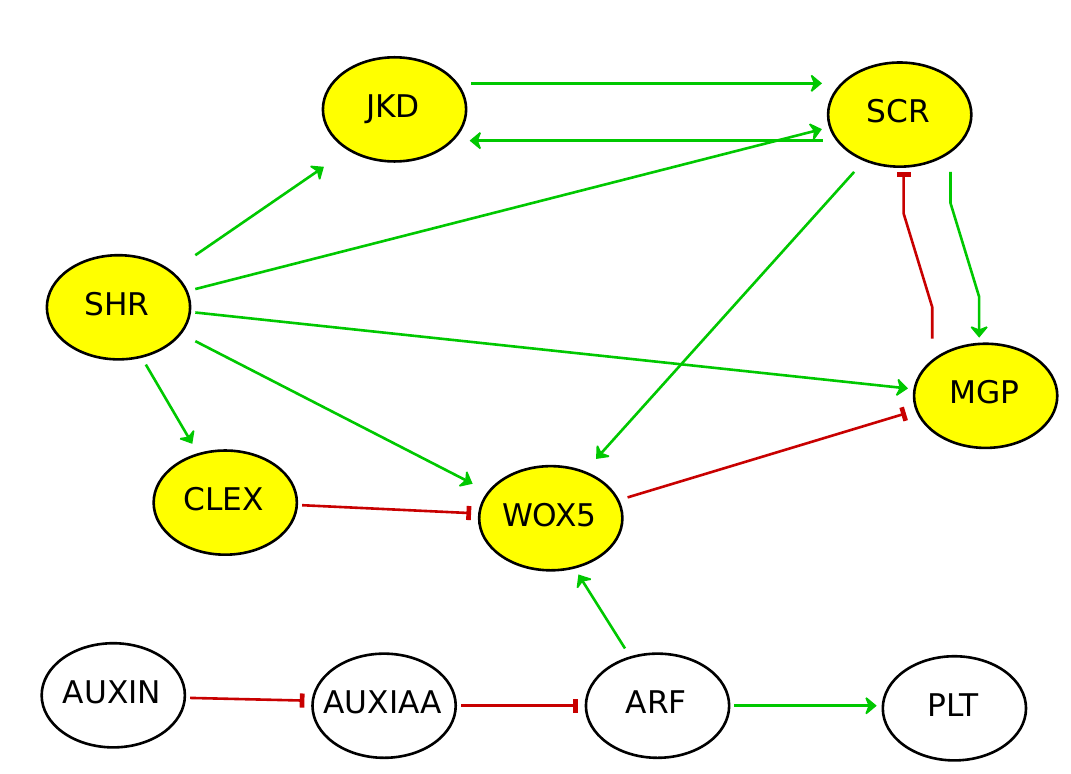}    
	\caption{
		Pictorial representation of the Boolean network from Example~\ref{ex:1} using GinSim~\cite{naldi2009logical}, adapted from~\cite{azpeitia2010single}.
	}
	\label{fig:celldif_DrosophilaBN}
\end{figure}

\subsection{Regulatory Networks}\label{sec:selCaseStudies}
We now apply GFB to Boolean and multi-valued networks from the literature.

\paragraph{\textbf{BN case study}} We study  the BN used as running example (Example~\ref{ex:1}). 
To ease interpretation, 
Fig.~\ref{fig:celldif_DrosophilaBN} uses the typical graphical notation of \emph{influence graphs} (offered, e.g.,  by GinSIM~\cite{naldi2009logical}). 
Nodes denote variables, while arrows denote \emph{influences} among nodes. Influences come from the update functions: green and red arrows denote, respectively, positive  (\emph{promotion}) and negative  (\emph{inhibition}) influence. 
In Example~\ref{ex:1},  $\arf$ promotes $\plt$ due to term $\arf$ in $f_\plt$, while $\auxin$ inhibits $\auxiaa$ due to term $\neg \auxin$ in $f_\auxiaa$.
The BN 
consists of two connected pathways: one for the transcription factor $\shr$ with its signalling to the other variables of the pathway (we highlight in yellow the involved nodes), and one  involving the hormone $\auxin$ and its signaling to the plethora ($\plt$) genes. 

BN variables can be categorized into three groups~\cite{naldi2012efficient}: \emph{inputs} ($\var{SHR}$, and $\var{AUXIN}$) that do not have incoming edges, \emph{outputs} ($\var{PLT}$) that do not have outgoing edges, and the remaining \emph{internal nodes}. 
The distinction is obvious from update functions: inputs have constant update functions, while outputs do not appear in the update function of other variables. Inputs are often set by the modeler to perform  \emph{what if} experiments, whereas outputs permit to observe the response dynamics of the model. In this BN, each input \emph{controls} its own pathway, meaning that the modeller can decide  to enable them via appropriate initial states.


Considering the GFB $\calX_R$ from Example~\ref{ex:2} for $\op = \land$, the only non-trivial block $C = \{ \scr, \shr, \jkd, \mgp, \wox, \clex \}$ 
corresponds to the yellow nodes in 
Fig.~\ref{fig:celldif_DrosophilaBN}.
This GFB is computed using 
the initial partition with two blocks separating outputs and non-output nodes. 
%
%
%
%
Considering the reduced model for $\calX_R$ from  Example~\ref{ex:3}, all yellow nodes in 
Fig.~\ref{fig:celldif_DrosophilaBN}
get collapsed into one, meaning that the $\shr$ pathway is abstracted away. 
In other words, in this example GFB has automatically identified and \emph{simplified} a pathway in the model, offering a coarser representation of the system focusing on the $\auxin$ pathway only.

\begin{figure}[t]
	\vspace{-1.3cm}
	\centering
	\includegraphics[width=0.88\linewidth]{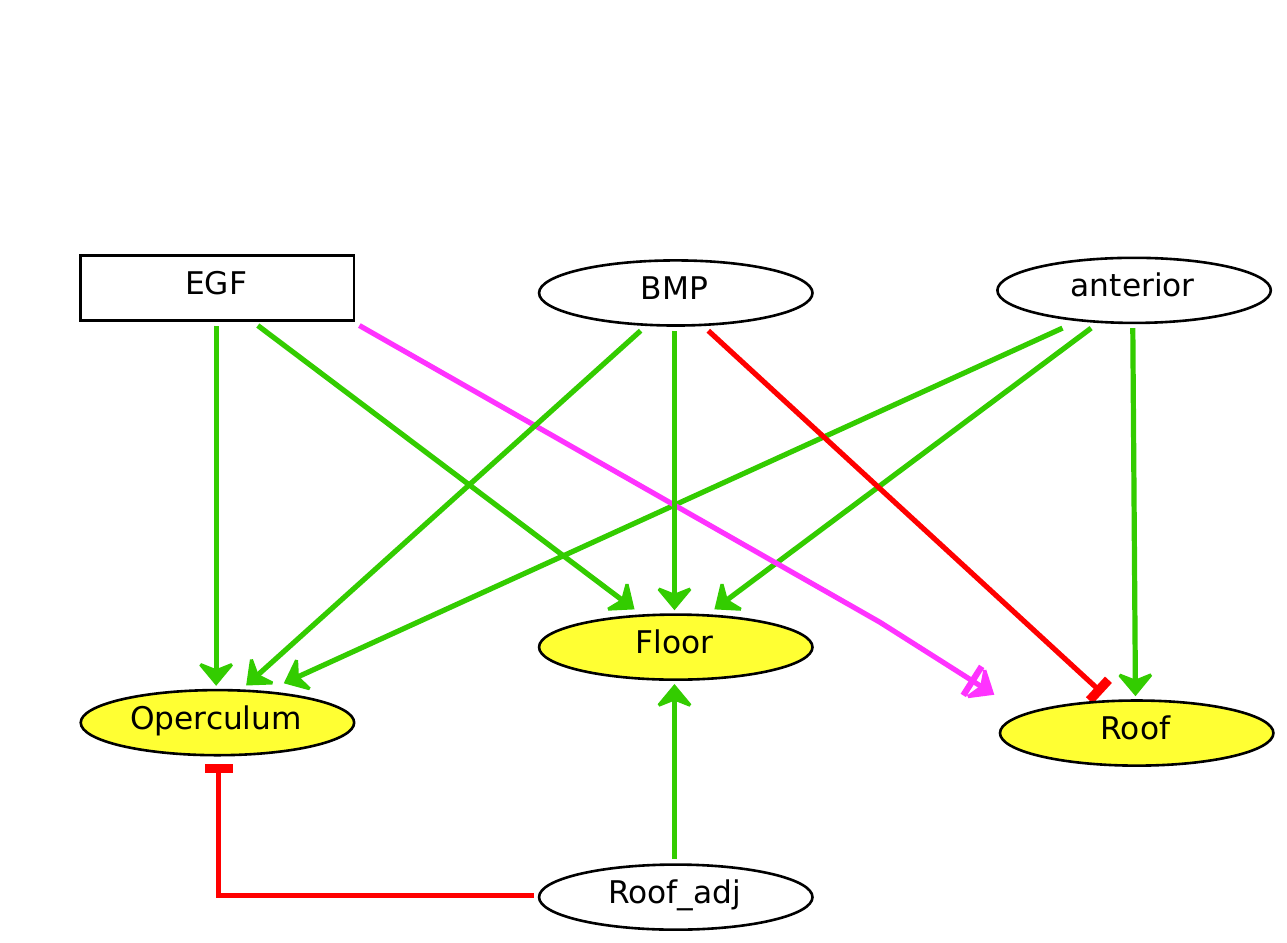}
	\caption{
		Pictorial representation using GinSim~\cite{naldi2009logical} adapted from~\cite{Drosophila} of the model on eggshell formation for drosophila melanogaster flies.}
	\label{fig:celldif_DrosophilaMV}
\end{figure}

\paragraph{\textbf{Multi-valued network case study}} We apply GFB to a multi-valued regulatory network (MV) from~\cite{Drosophila}.
Intuitively,  an MV is a BN where variables can admit more than two values. This is a single-cell model describing the development of eggshell structures in drosophila melanogaster flies. 
The MV has 7 variables with relations depicted in Fig.~\ref{fig:celldif_DrosophilaMV} and 
update functions:

\begin{align*}
	f_\var{EGF} & = \var{EGF} \\
	f_\var{BMP} & = \var{BMP} \nonumber\\
	f_\var{Ant} & = \var{Ant} \nonumber\\
	f_\var{RoofAdj} & = \var{RoofAdj}\nonumber\\
	f_\var{Roof} & = \var{Ant}\!:\!\!1 \land \var{EGF}\!:\!\!1 \land \var{BMP}\!:\!\!0 \\
	f_\var{Floor} & = \var{Ant}\!:\!\!1 \land (\var{EGF}\!:\!2 \lor (\var{EGF}\!:\!\!1 \land \var{BMP}\!:\!\!1 ) ) \land \var{RoofAdj}\!:\!\!1 \nonumber\\
	f_\var{Operc} & = \var{Ant}\!:\!\!1 \land (\var{EGF}\!:\!2 \lor (\var{EGF}\!:\!\!1 \land \var{BMP}\!:\!\!1 ) ) \land \var{RoofAdj}\!:\!0 \nonumber
\end{align*}
Using the notation in~\cite{Drosophila}, ``$\var{var}\!:\!v$'' stands for \emph{variable $\var{var}$ has value $v$}. This is a Boolean predicate evaluating to 1 if $\var{var}$ has value $v$, and $0$ otherwise.  
Variable $\var{EGF}$, the rectangular node in 
Fig.~\ref{fig:celldif_DrosophilaMV},
can take values 0, 1, 2, denoting absent/intermediate/high  activation levels. All other variables 
are Boolean (0/1).\footnote{
Our framework requires all variables to have same domain $\MM$.  In order to support MV, we implicitly \emph{expand} the domain of all variables to the largest one (e.g., $\{0,1,2\}$ of $\var{EGF}$). This does not change the models' dynamics, in the sense that when setting initial states fitting in the original domain we will remain within the original domain.
}

Differently from Fig.~\ref{fig:celldif_DrosophilaBN},
variables divide in two groups only: the \emph{inputs} $\var{EGF}$, $\var{BMP}$, $\var{Ant}$, and $\var{RoofAdj}$, and the 
\emph{outputs} $\var{Operc}$, $\var{Floor}$, and $\var{Roof}$. 
We also have a third edge type, the purple one from $\var{EGF}$ to $\var{Roof}$. This visually stresses that $\var{EGF}$ influences $\var{Roof}$ only when in intermediate level and not when in high level. 

\begin{figure*}[t]
	\centering
	\vspace{-0.3cm}
	\includegraphics[width=0.5\linewidth]{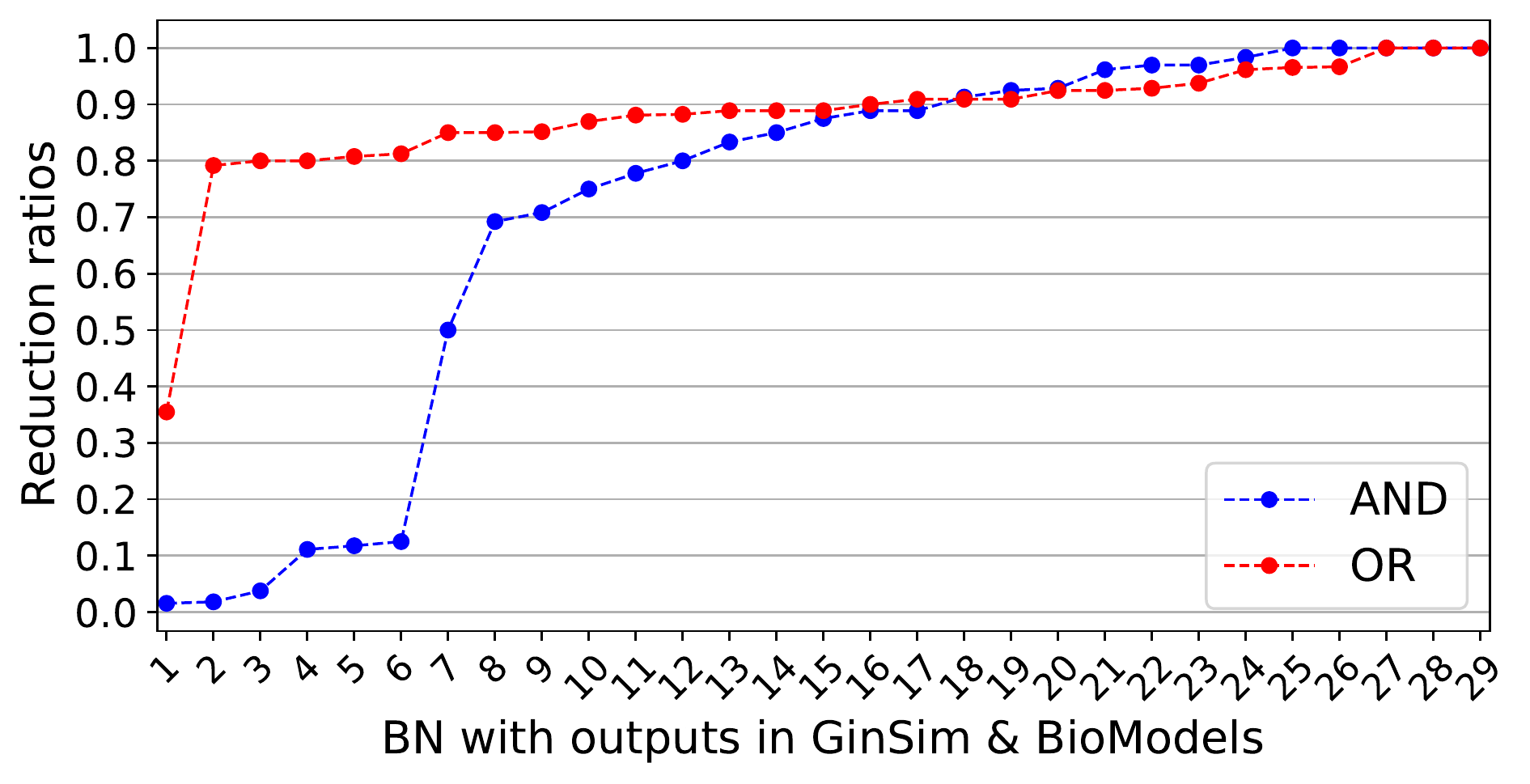}\hspace{-0.2cm}
	\includegraphics[width=0.5\linewidth]{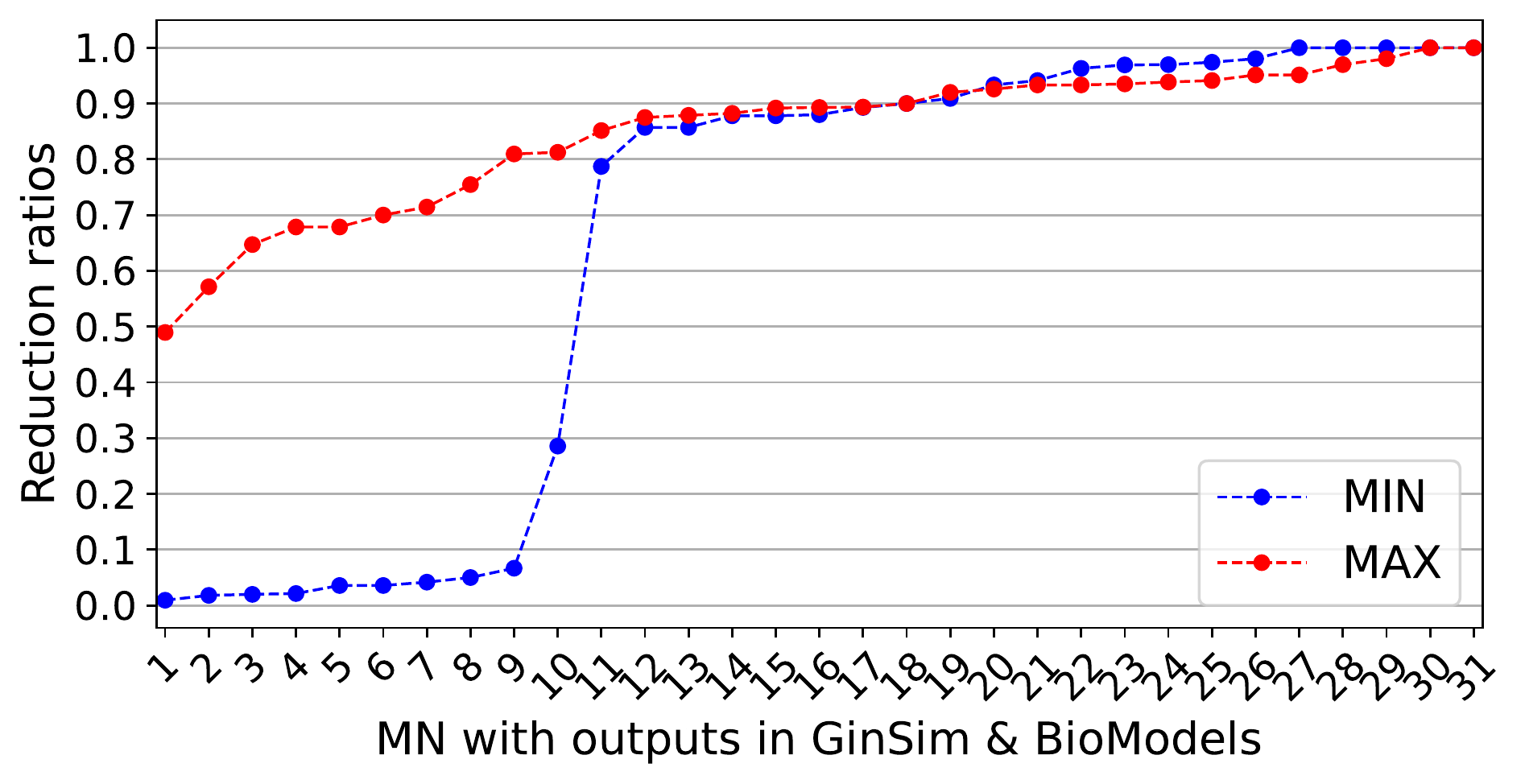}\\
	\caption{
		\textbf{(Left)}
		Reduction ratios (reduced variables over original ones) in ascending order for the 29 BN with outputs from GINsim and BioModelsDB for  $\oplus\in\{\land,\lor\}$ and initial partitions with two blocks separating output and non-outputs.
		\textbf{(Right)}
		Same as (Left) for the 31 MV with outputs from the two repositories using  $\oplus\in\{\min,\max\}$.
		%
	}
	\label{BNcharts_MNcharts_attr_WN_redratio}
\end{figure*}

\begin{figure}[t]
	\centering
	\includegraphics[width=1\linewidth]{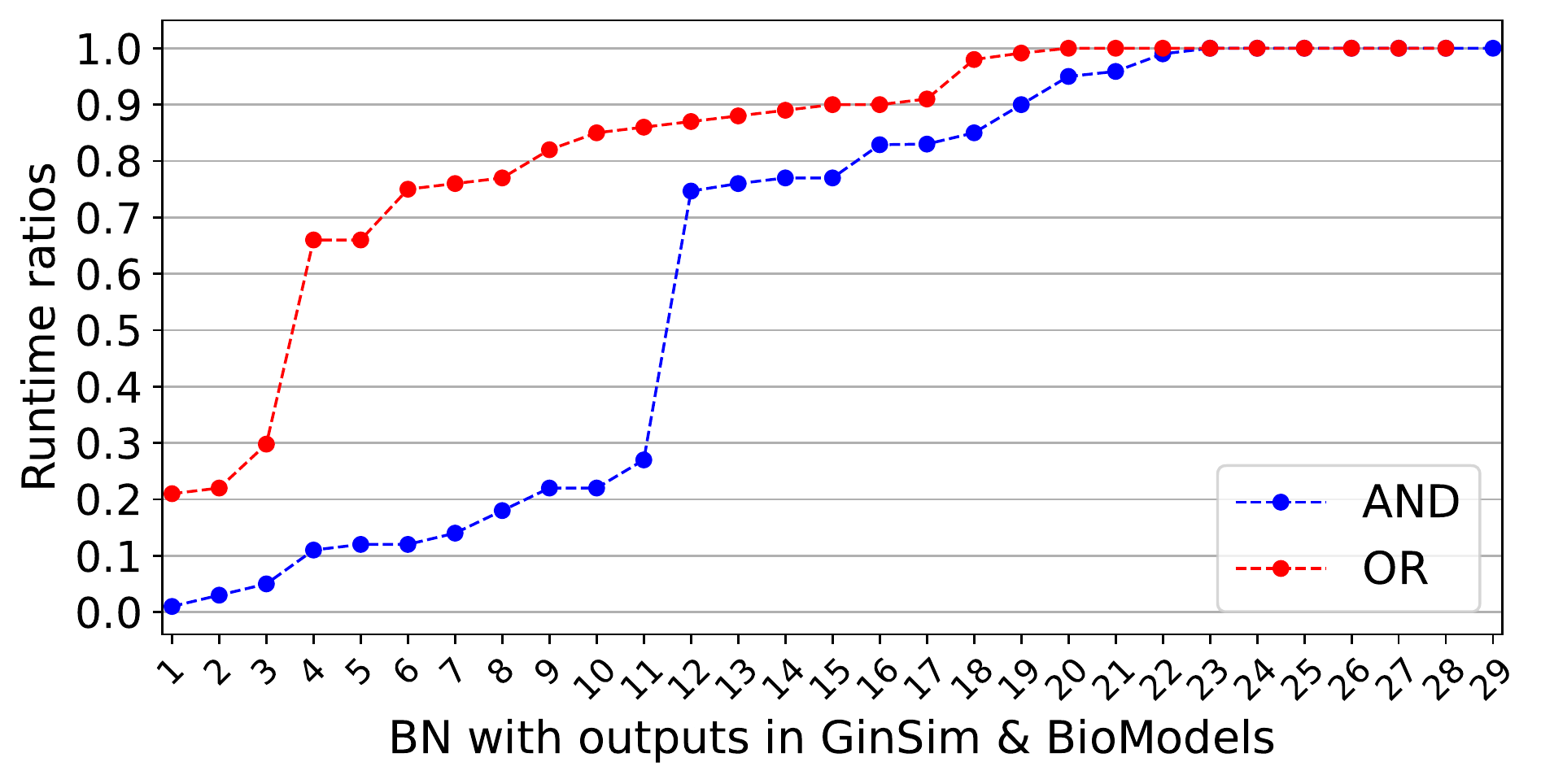}	
	\vspace{-0.6cm}
	\caption{
		Runtime ratios in ascending order for  computation of attractors for the 29 BNs from Fig.~\ref{BNcharts_MNcharts_attr_WN_redratio} (Left) and their reductions.
	}
	\label{BNcharts_MNcharts_attr_WN_runtime}
\end{figure}

The MV relates three  follicle cell fates, the outputs, 
to combinations of values of the  inputs. 
EGF and BMP are known 
 pathways responsible for patterning of the drosophila eggshell~\cite{Drosophila}. This is encoded in the model because $\var{EGF}$ and $\var{BMP}$ influence, in different ways, all outputs. 
Finally, $\var{Ant}$ 
models the anterior competence region, therefore it is required by all outputs, while $\var{RoofAdj}$ 
accounts for the state of neighboring cells by promoting $\var{Floor}$ and inhibiting $\var{Operc}$ (\emph{operculum}).

The partition with one block for all outputs 
and singleton blocks for each input is a GFB for $\oplus\in\{\max,\min\}$. By Definition~\ref{def:red}, we get two different reduced models in the two cases, enabling complementary studies.
Case $\max$ allows \emph{full output deactivation} studies, meaning that the reduced variable for the outputs gets value 0 only if all outputs have value 0. Instead,  case $\min$ allows \emph{full output activation} studies, as the reduced variable gets value 1 only when all outputs have value 1.
By naming $\var{outputs}$ the reduced variable corresponding to the block of outputs, by applying Definition~\ref{def:red} and some algebraic simplification we get:
\begin{align*}
	f_\var{outputs}  &= \var{Ant}\!:\!\!1 \land (\var{EGF}\!:\!\!1 \lor \var{EGF}\!:\!\!2),  \ \ & \text{for $\oplus=\max$,} \\
	f_\var{outputs}  &= 0, & \text{for $\oplus= \min$,}
\end{align*}
while the update functions of the input variables remain unchanged.
From this we get that: despite the three outputs have different dependencies on $\var{Ant}$, $\var{BMP}$, $\var{RoofAdj}$, and on different values of $\var{EGF}$, in the $\oplus=\max$ case it is enough to consider only $\var{ANT}$ and $\var{EGF}$ to answer questions related to full output deactivation. Furthermore, it is not necessary anymore to use three values for $\var{EGF}$, as we are only interested in the cases in which it is $0$ or positive $(\var{EGF}\!:\!\!1 \lor \var{EGF}\!:\!\!2)$.
Instead, from the $\oplus=\min$ case we know that the original model never expresses cases of full activation, i.e., it never happens that the three outputs have all value 1. Indeed, 
by studying 
the update functions of the original outputs, we see that there are no 
values for the involved variables that makes all of them true.





\paragraph{\textbf{Large-scale validation of GFB on regulatory networks}} We present a large-scale validation of GFB on the BNs and MVs from the repositories GinSim~(\url{ginsim.org/models\_repository}) and BioModelsDB~\cite{BioModels2020}.~\footnote{BioModelsDB contains both BN and ODE models. Here we focus on BNs, while ODEs were considered in~\cite{cmsb2019tcs} for ODE-based reduction techniques.} We validate GFB in terms of aggregation power and of speed-up offered for attractor analysis.

\emph{Experimental setting.} 
We created a prototype implementation of GFB integrated with the SMT solver Z3~\cite{de2008z3} to check formulas from $\Psi^\calH_{x_i,x_j}$ in Algorithm~\ref{algorithm} in the tool for model reduction ERODE~\cite{cardelli2017erode}, which has been recently extended with support for BNs~\cite{DBLP:conf/cmsb/ArgyrisLTTV22}. 
%
In doing so, we added to ERODE 
an importer for SBML Qual~\cite{chaouiya2013sbml}, an XML format 
supported by both repositories, allowing us to import all 43 
BNs and
50 
MVs. 
In order to obtain physically-relevant initial partitions, 
we infer \emph{candidate outputs}, variables not appearing in the update function of other variables. For each model, we create \emph{output-preserving} initial partitions: these consist of two blocks, one containing all outputs and one containing the remaining variables.   This guarantees that reduced models allow, e.g., for  full output (de)activation studies discussed before.
In order to perform a consistent 
treatment, we restricted our analysis on the 29 BNs and 31 MVs with at least one candidate output. 

\emph{Validation of aggregation power.} Fig.~\ref{BNcharts_MNcharts_attr_WN_redratio} (Left) provides the reduction ratios obtained for the BNs  using $\oplus\in\{\land, \lor\}$. For each model we plot the reduction ratio, defined as the number of reduced variables over that  of original ones. 
For each operator $\oplus$, the ratios were sorted in ascending order.
We can see that $\oplus=\land$ has high aggregation power, with about one fourth of the models having reduction ratio below 0.6, while for $\oplus=\lor$ most of the models have 0.8 or more.
For $\oplus=\land$, some models have particularly low ratios, below 0.2, some of which due to the fact that the reduced model has 2 variables only. We remark that these shall not be considered  \emph{degenerate} reductions, because of the used initial partitions, as discussed.  
We do not present results on maximal reductions, 
obtained with the initial partition with one block only. These are significantly smaller, but some are degenerate with one variable only.
We leave for future work a detailed study on finer intermediate reductions using model-specific initial partitions preserving 
variables of interest for the modeler. For example, a modeler could be interested in preserving only some outputs. Fig.~\ref{BNcharts_MNcharts_attr_WN_redratio} (Right) presents a similar study performed on the MVs using $\oplus=\min$ and $\oplus\!=\!\max$, confirming the aggregation power of GFB.

\emph{Validation of analysis speed-up.}
Corollary~\ref{cor:attractor} ensures that GFB  maps all attractors of the original system to attractors of the reduced one. Here we show that
this can speed-up attractor computation. We use the COLOMOTO Notebook~\cite{naldi2015cooperative}, an environment  incorporating a variety of tools for BN analysis. An example is
BNS~\cite{dubrova2011sat}, which combines SAT-solving and bounded model checking to identify attractors. We computed the attractors of the 29 considered BNs and of their reductions. We could not consider MVs because we are not aware of tools for general attractor analysis for MVs. 
Fig.~\ref{BNcharts_MNcharts_attr_WN_runtime} shows the obtained runtime ratios (computation time of attractors in the reduced model over that in the original one).
%
%
In several cases the reduction led to significant analysis speed-ups: in 11 BNs the ratio is less than 0.3. We remark that GFB is \emph{useful}, because the analysis of the original BNs, the AND- and OR-reductions 
took on average 100s, 30s and 60s, respectively. Notably, reductions with low reduction ratios are particularly fast (fewer algorithm iterations): the 6 AND-reductions in Fig.~\ref{BNcharts_MNcharts_attr_WN_redratio} (Left) with ratio smaller than 0.3 take less than 1.5 seconds on average. 

\begin{table}[b]
	\centering
	\begin{tabular}{c c c c }
		\toprule
		\emph{Model} &  \emph{Variables} & \multicolumn{2}{c}{\emph{Attractors analysis}}  \\
		&   & \emph{Count} & \emph{Runtime(s)}  
		\\		
		\midrule
		\emph{Original}& 321	&  \multicolumn{2}{c}{---\emph{Time Out}---} 
		\\
		\midrule
		\emph{Output separated}& 189	&  \multicolumn{2}{c}{---\emph{Time Out}---} 
		\\
		\emph{O1} &70	&64	&0.668 
		\\
		\emph{O2} &33	&64	&0.325 
		\\
		\emph{Maximal} &1	&1	&0.001 
		\\
		\bottomrule
	\end{tabular}
\vspace{0.2cm}
	\caption{GFB enables attractors computation on  large BN~\cite{raza2008logic}.}
	\label{res}
\end{table}

\paragraph{\textbf{Enabling analysis of large BNs using GFB}}
We now apply GFB to a large BN of signalling pathways central to macrophage activation~\cite{raza2008logic}. This BN contains $\mathit{321}$ variables, 
making attractor computation infeasible even using the most efficient tool for this task~\cite{dubrova2011sat}. In particular, the analysis does not terminate within an arbitrarily chosen time limit of 10 hours. 
Our crucial hypothesis is that \emph{GFB can enable some analysis} of this otherwise not analyzable BN, although with certain restrictions imposed by what is exactly preserved by the reduction. 

The results are presented in Table~\ref{res}. In this experiment we focus on $\oplus=\wedge$.
We can see that the maximal reduction is not physically-relevant, as it reduces to 1 variable only.
The output-preserving reduction, instead, leads to a reduced model with 189 variables. 
Despite this, the obtained reduced model is still not analyzable within the chosen time limit. 
We now show how two alternative initial partitions lead to reduced models that can be effectively analyzed. 
In particular, we assume that the modeler is not interested in preserving all 68 outputs, but two different subsets of them: $O_1=\{\var{S\_28}, \var{S\_26}, \var{S\_198}, \var{S\_11}\}$ and $O_2=\{\var{S\_184}, \var{S\_188}\}$. In both cases, we use an initial partition with one block for the selected outputs, and one for all the other variables.
In these two cases, we obtained models with 70 and 33 variables, respectively, which admit analysis. In particular, the obtained reduced models can now be analyzed using less than a second.

\subsection{Non-linear reductions of Differential and Difference Equations}
We present examples of  exact lumping where $\psi_R$ is not linear, and thus cannot be captured by linear lumpings  such FDE. 
We use  $(\RE,\cdot)$ with  neutral element $1$.

\input{lotka}
\input{weightedNetworks}

%% file: lotka.tex
\paragraph{\textbf{Nonlinear Reduction of a Lotka-Volterra Model over $(\RE,\cdot)$}}\label{sec:lotka}
%
We start considering a prototypical 
higher-order Lotka-Volterra model~\cite{PredatorPreyHO} where  $x_1$ preys $x_2$ and $x_3$, while $x_2$ and $x_3$ prey together $x_1$. The corresponding ODE system is
\begin{equation}
\begin{split}
	\label{eq:predator}
	\partial_t v_{x_1}  &= v_{x_1} (1 - v_{x_2} v_{x_3}), \\ 
	\partial_t v_{x_2}  & = v_{x_2} (1 - v_{x_1}), \\
	\partial_t v_{x_3} & = v_{x_3} (1 - v_{x_1}).
\end{split}
\end{equation}
The ODE discretization of~(\ref{eq:predator}) is given by
\begin{align*}
 f_{x_1}(s) &= s_{x_1} + \dt s_{x_1} (1 \!-\! s_{x_2} s_{x_3}), \\
f_{x_2}(s) & = s_{x_2} + \dt s_{x_2} (1 \!-\! s_{x_1}), \\
f_{x_3}(s)  &= s_{x_3} + \dt s_{x_3} (1 \!-\! s_{x_1}).
\end{align*}

By Theorem~\ref{thm:cont}, the \emph{nonlinear} function $\psi_R(v_{x_1},v_{x_2},v_{x_3}) = (v_{x_1}, v_{x_2} \cdot v_{x_3})$ is an exact lumping of~(\ref{eq:predator}). Indeed, $\calX_R \!=\! \{ \{x_1\}, \{x_2,x_3\} \}$ is a GFB of~(\ref{eq:predator}) for $\op=\cdot$ because $\Psi^{\calX_R}_{x_2,x_3}$ is  valid thanks to the identities 
$f_{x_1} = f_{x_1}[x_2 / 1, x_3 / x_2 x_3]$, and
\begin{align*}
	f_{x_2} \cdot f_{x_3} & = (x_2 + \dt x_2 (1 - x_1)) \cdot (x_3 + \dt x_3 (1 - x_1)) \nonumber \\
	& = x_2 x_3 + 2 \dt x_2 x_3 (1 - x_1) + \dt^2 x_2 x_3 (1 - x_1)^2\\
	& =  (f_{x_2} \cdot f_{x_3})[x_2 / 1, x_3 / x_2 x_3] . \nonumber
\end{align*}
The lumped ODE system is given by $\partial_t v_{x_1} = v_{x_1} (1 - v_{x_2} v_{x_3})$ and
\begin{align*}
\partial_t (v_{x_2} v_{x_3}) &= \partial_t v_{x_2} \cdot v_{x_3} + v_{x_2} \cdot \partial_t v_{x_3} \\
&= v_{x_2} (1 - v_{x_1}) v_{x_3} + v_{x_2} v_{x_3} (1 - v_{x_1}) \\
&= 2  v_{x_1} v_{x_2} (1 - v_{x_1}).
\end{align*}

%% file: weightedNetworks.tex
\begin{figure}[t]
	\centering
	\includegraphics[width=0.98\linewidth]{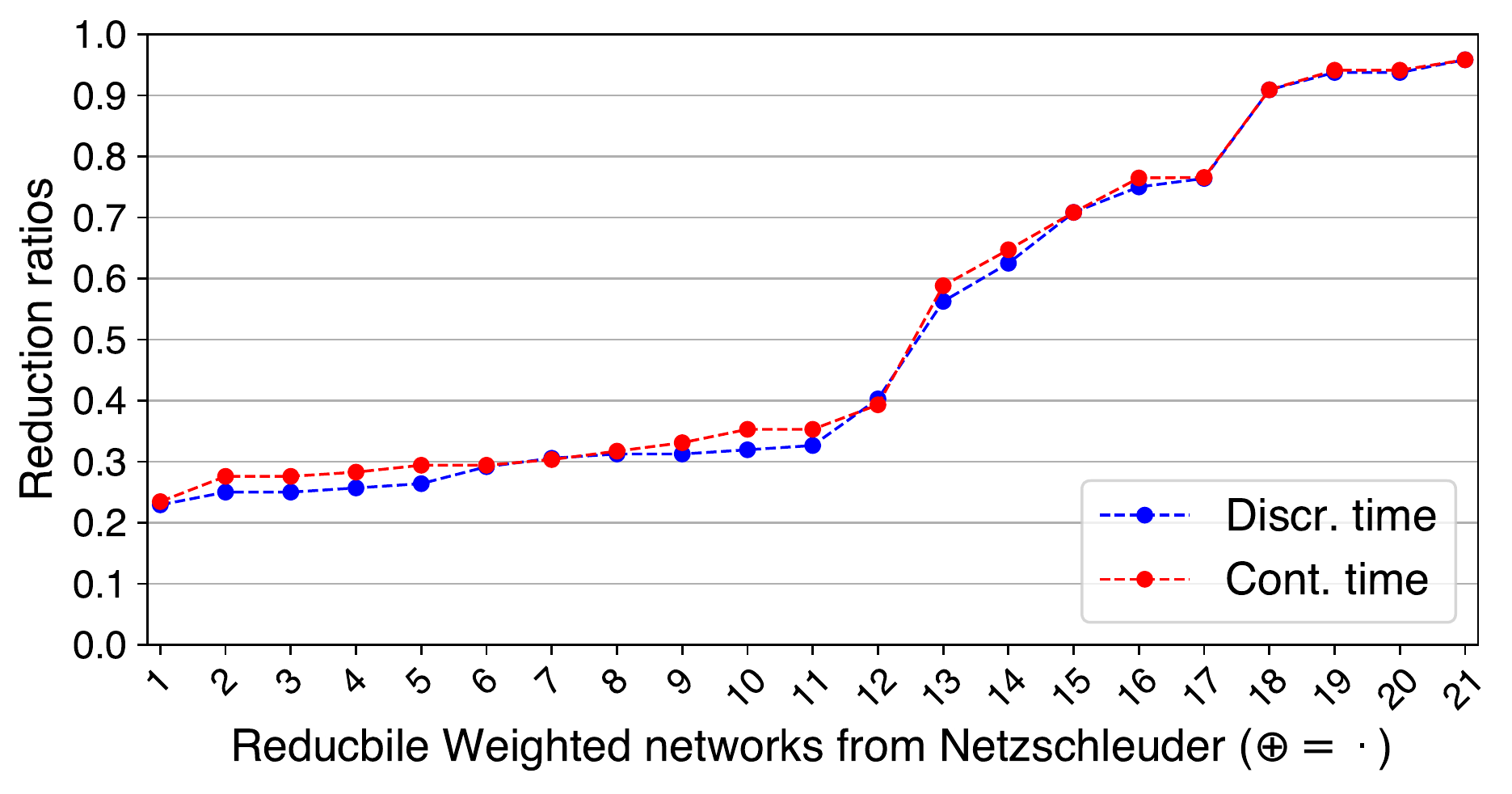}	
 \vspace{-0.2cm}
	\caption{
		Similarly to Fig.~\ref{BNcharts_MNcharts_attr_WN_redratio}, we plot reduction ratios for the discrete- and continuous-time dynamical interpretations of 21 weighted networks from Netzschleuder. We used one operator, $\oplus=\cdot$, and initial partitions separating the first node in the network from the others. 
	}
	\label{BNcharts_MNcharts_attr_WN_Network}
\end{figure}	

\paragraph{\textbf{Nonlinear Reduction of Dynamical Weighted Networks over $(\RE,\cdot)$}}
We now consider real-valued DS obtained from weighted networks from the  Netzschleuder repository~\cite{Netzschleuder}. 
We considered all 72 weighted networks with at most 200 nodes  (by restricting to at most the first 15 models from each family of models). The undirected ones were expanded in directed by replacing every undirected edge with two corresponding directed ones with same weight.

We consider two different dynamical interpretations. For $A$ the adjacency matrix of a network, we study the discrete-time DS $x(t+1)=Ax(t)$, and the (ODE discretization of the) continuous-time DS $\partial_t v(t)=Av(t)$. 
In both cases, we use one variable per node.\footnote{In the continuous-time case, we also have an additional variable for the $\tau$ term from the ODE discretization, to which we give constant update function. This guarantees that the obtained reductions hold for any value of $\tau$.} 
%
We use $\oplus=\cdot$, obtaining nonlinear reductions. 
This leads to high nonlinearities in formulas $\Psi^{\calX_R}_{x_i,x_j}$ from Theorem~\ref{thm:bin:fe}, complex to handle for Z3. Indeed,  Algorithm~\ref{algorithm} failed to terminate within an arbitrarily chosen time-out of 1 hour even for models of moderate size. 
Hence, for our experiments we used the randomized version of the algorithm discussed in Section~\ref{sec:comp}, performing 40 tests per formula $\Psi^{\calX_R}_{x_i,x_j}$ from Theorem~\ref{thm:bin:fe} after sampling  values for all variables. Currently, our prototype is still based on Z3, to which we provide  the sampled values making all $\Psi^{\calX_R}_{x_i,x_j}$ formulas variable free. In this setting, Z3 never failed.

Fig.~\ref{BNcharts_MNcharts_attr_WN_Network} provides the results for the 21 networks that admitted a reduction, 29\% of the 72 considered. 
We got similar reduction ratios in the two interpretations, with slightly better ones for the discrete-time one.
The lower reduction power of the 
continuous-time case comes from two factors: (i) Models have higher nonlinearities due to the $\tau$ term; (ii) Theorem~\ref{thm:cont} gives only a necessary condition for aggregation in this case (our  prototype does not support the results of Theorem~\ref{thm:cont:char}).
The largest runtimes for the continuous- and discrete-time cases were  about 500 and 400 seconds, respectively, for a model with 145 nodes.

%% file: conlcusions.tex
\section{Conclusion}\label{sec:conc}

Generalized forward bisimulation (GFB) is a technique for dimensionality reduction of discrete- and continuous-time dynamical systems that captures and generalizes existing techniques. 
GFB allows to compute nonlinear reductions. One needs to specify a dynamical system,  
a commutative monoid (the variables' domain and an operation used to aggregate them), and an initial partition of the variables (used to tune the reduction power to preserve variables of interest). 
A partition refinement algorithm then minimizes the system over the operation of the monoid. We implemented GFB 
and applied it to four popular formalisms: difference and differential equations with monoid $(\RE,\cdot)$, Boolean networks with $(\BB,\land)$ and $(\BB,\lor)$, multi-valued networks with $(\{0,1,2\},\min)$ and $(\{0,1,2\},\max)$. In all cases, GFB yielded notable nonlinear reductions. On 60 Boolean and multi-valued networks from two popular repositories, we have shown high aggregation power  and analysis speed-ups. Using an existing large Boolean network with 321 variables we have shown that GFB might enable the analysis of otherwise untractable models. On 21 ODEs originated from weighted networks from a popular repository, we have computed nonlinear reductions thanks to the $\cdot$ operation, showing high aggregation power. Future work will study the reduction of optimization problems from systems biology~\cite{DBLP:journals/tac/WhitbyCKLTT22}, performance engineering~\cite{DBLP:journals/tcad/0001MTA16} and AI~\cite{DBLP:journals/corr/abs-2106-13898}.

\emph{\textbf{Acknowledgments.}} The work was partially supported by the DFF project REDUCTO 9040-00224B, the Poul Due Jensen Grant 883901, the Villum Investigator Grant S4OS, and the PRIN project SEDUCE 2017TWRCNB.

%% file: ms.bbl
\begin{thebibliography}{10}
\providecommand{\url}[1]{#1}
\csname url@samestyle\endcsname
\providecommand{\newblock}{\relax}
\providecommand{\bibinfo}[2]{#2}
\providecommand{\BIBentrySTDinterwordspacing}{\spaceskip=0pt\relax}
\providecommand{\BIBentryALTinterwordstretchfactor}{4}
\providecommand{\BIBentryALTinterwordspacing}{\spaceskip=\fontdimen2\font plus
\BIBentryALTinterwordstretchfactor\fontdimen3\font minus
  \fontdimen4\font\relax}
\providecommand{\BIBforeignlanguage}[2]{{%
\expandafter\ifx\csname l@#1\endcsname\relax
\typeout{** WARNING: IEEEtran.bst: No hyphenation pattern has been}%
\typeout{** loaded for the language `#1'. Using the pattern for}%
\typeout{** the default language instead.}%
\else
\language=\csname l@#1\endcsname
\fi
#2}}
\providecommand{\BIBdecl}{\relax}
\BIBdecl

\bibitem{sangiorgi:book}
D.~Sangiorgi, \emph{Introduction to Bisimulation and Coinduction}.\hskip 1em
  plus 0.5em minus 0.4em\relax Cambridge University Press, 2011.

\bibitem{DBLP:conf/tcs/Park81}
D.~Park, ``Concurrency and automata on infinite sequences,'' in
  \emph{Theoretical Computer Science}, 1981, pp. 167--183.

\bibitem{Larsen19911}
K.~G. Larsen and A.~Skou, ``Bisimulation through probabilistic testing,''
  \emph{Inf. Comput.}, vol.~94, no.~1, pp. 1--28, 1991.

\bibitem{BuchholzOrdinaryExact}
P.~Buchholz, ``Exact and ordinary lumpability in finite {M}arkov chains,''
  \emph{Journal of Applied Probability}, vol.~31, no.~1, pp. 59--75, 1994.

\bibitem{De-Nicola:1990aa}
R.~De~Nicola, U.~Montanari, and F.~Vaandrager,
  ``\BIBforeignlanguage{English}{Back and forth bisimulations},'' in
  \emph{\BIBforeignlanguage{English}{CONCUR '90 Theories of Concurrency:
  Unification and Extension}}, ser. Lecture Notes in Computer Science,
  J.~Baeten and J.~Klop, Eds.\hskip 1em plus 0.5em minus 0.4em\relax Springer
  Berlin Heidelberg, 1990, vol. 458, pp. 152--165.

\bibitem{concur15}
L.~Cardelli, M.~Tribastone, M.~Tschaikowski, and A.~Vandin, ``Forward and
  backward bisimulations for chemical reaction networks,'' in \emph{{CONCUR}},
  2015, pp. 226--239.

\bibitem{Bioinf21}
L.~Cardelli, I.~C. P{\'{e}}rez{-}Verona, M.~Tribastone, M.~Tschaikowski,
  A.~Vandin, and T.~Waizmann, ``Exact maximal reduction of stochastic reaction
  networks by species lumping,'' \emph{Bioinform.}, vol.~37, no.~15, pp.
  2175--2182, 2021.

\bibitem{DBLP:conf/popl/CardelliTTV16}
L.~Cardelli, M.~Tribastone, M.~Tschaikowski, and A.~Vandin, ``Symbolic
  computation of differential equivalences,'' in \emph{{POPL}}, 2016, pp.
  137--150.

\bibitem{paige1987three}
R.~Paige and R.~E. Tarjan, ``Three partition refinement algorithms,''
  \emph{SIAM Journal on Computing}, vol.~16, no.~6, pp. 973--989, 1987.

\bibitem{okino1998}
M.~S. Okino and M.~L. Mavrovouniotis, ``Simplification of mathematical models
  of chemical reaction systems,'' \emph{Chemical Reviews}, vol.~2, no.~98, pp.
  391--408, 1998.

\bibitem{DBLP:conf/splc/Tribastone14}
M.~Tribastone, ``Behavioral relations in a process algebra for variants,'' in
  \emph{{SPLC}}, S.~Gnesi, A.~Fantechi, P.~Heymans, J.~Rubin, K.~Czarnecki, and
  D.~Dhungana, Eds.\hskip 1em plus 0.5em minus 0.4em\relax {ACM}, 2014, pp.
  82--91.

\bibitem{antoulas}
A.~Antoulas, \emph{Approximation of Large-Scale Dynamical Systems}, ser.
  Advances in Design and Control.\hskip 1em plus 0.5em minus 0.4em\relax SIAM,
  2005.

\bibitem{Snowden:2017aa}
T.~J. Snowden, P.~H. van~der Graaf, and M.~J. Tindall, ``Methods of model
  reduction for large-scale biological systems: A survey of current methods and
  trends,'' \emph{Bulletin of Mathematical Biology}, vol.~79, no.~7, pp.
  1449--1486, 2017.

\bibitem{LI1994343}
G.~Li, H.~Rabitz, and J.~T{\'o}th, ``A general analysis of exact nonlinear
  lumping in chemical kinetics,'' \emph{Chemical Engineering Science}, vol.~49,
  no.~3, pp. 343--361, 1994.

\bibitem{KAUFFMAN1969437}
S.~Kauffman, ``Metabolic stability and epigenesis in randomly constructed
  genetic nets,'' \emph{Journal of Theoretical Biology}, vol.~22, no.~3, pp.
  437 -- 467, 1969.

\bibitem{PNAScttv}
L.~Cardelli, M.~Tribastone, M.~Tschaikowski, and A.~Vandin, ``Maximal
  aggregation of polynomial dynamical systems,'' \emph{Proceedings of the
  National Academy of Sciences}, vol. 114, no.~38, pp. 10\,029 -- 10\,034,
  2017.

\bibitem{hopfensitz2013attractors}
M.~Hopfensitz, C.~M{\"u}ssel, M.~Maucher, and H.~A. Kestler, ``Attractors in
  boolean networks: a tutorial,'' \emph{Computational Statistics}, vol.~28,
  no.~1, pp. 19--36, 2013.

\bibitem{DBLP:conf/qest/CardelliTTV18}
L.~Cardelli, M.~Tribastone, M.~Tschaikowski, and A.~Vandin, ``Guaranteed error
  bounds on approximate model abstractions through reachability analysis,'' in
  \emph{QEST}, 2018, pp. 104--121.

\bibitem{cardelli2017erode}
------, ``Erode: a tool for the evaluation and reduction of ordinary
  differential equations,'' in \emph{International Conference on Tools and
  Algorithms for the Construction and Analysis of Systems}.\hskip 1em plus
  0.5em minus 0.4em\relax Springer, 2017, pp. 310--328.

\bibitem{thomas1995dynamical}
R.~Thomas, D.~Thieffry, and M.~Kaufman, ``Dynamical behaviour of biological
  regulatory networks --- i. biological role of feedback loops and practical
  use of the concept of the loop-characteristic state,'' \emph{Bulletin of
  mathematical biology}, vol.~57, no.~2, pp. 247--276, 1995.

\bibitem{argyris2021reducing}
G.~Argyris, A.~Lluch~Lafuente, M.~Tribastone, M.~Tschaikowski, and A.~Vandin,
  ``Reducing boolean networks with backward boolean equivalence,'' in
  \emph{International Conference on Computational Methods in Systems
  Biology}.\hskip 1em plus 0.5em minus 0.4em\relax Springer, 2021, pp. 1--18.

\bibitem{NALDI2009134}
A.~Naldi, D.~Berenguier, A.~Fauré, F.~Lopez, D.~Thieffry, and C.~Chaouiya,
  ``Logical modelling of regulatory networks with ginsim 2.3,''
  \emph{Biosystems}, vol.~97, no.~2, pp. 134--139, 2009.

\bibitem{BioModels2020}
\BIBentryALTinterwordspacing
R.~S. Malik-Sheriff, M.~Glont, T.~V.~N. Nguyen, K.~Tiwari, M.~G. Roberts,
  A.~Xavier, M.~T. Vu, J.~Men, M.~Maire, S.~Kananathan, E.~L. Fairbanks, J.~P.
  Meyer, C.~Arankalle, T.~M. Varusai, V.~Knight-Schrijver, L.~Li,
  C.~Dueñas-Roca, G.~Dass, S.~M. Keating, Y.~M. Park, N.~Buso, N.~Rodriguez,
  M.~Hucka, and H.~Hermjakob, ``{BioModels — 15 years of sharing
  computational models in life science},'' \emph{Nucleic Acids Research},
  vol.~48, no.~D1, pp. D407--D415, 1 2020, gkz1055. [Online]. Available:
  \url{https://doi.org/10.1093/nar/gkz1055}
\BIBentrySTDinterwordspacing

\bibitem{PredatorPreyHO}
P.~Singh and G.~Baruah, ``Higher order interactions and species coexistence,''
  \emph{Theoretical Ecology}, vol.~14, pp. 71--83, 2021.

\bibitem{Netzschleuder}
\BIBentryALTinterwordspacing
T.~P. Peixoto, ``The netzschleuder network catalogue and repository,'' 2020.
  [Online]. Available: \url{https://networks.skewed.de/}
\BIBentrySTDinterwordspacing

\bibitem{Li19891413}
G.~Li and H.~Rabitz, ``A general analysis of exact lumping in chemical
  kinetics,'' \emph{Chemical Engineering Science}, vol.~44, no.~6, pp.
  1413--1430, 1989.

\bibitem{tomlin1997effect}
A.~S. Tomlin, G.~Li, H.~Rabitz, and J.~T{\'o}th, ``The effect of lumping and
  expanding on kinetic differential equations,'' \emph{SIAM Journal on Applied
  Mathematics}, vol.~57, no.~6, pp. 1531--1556, 1997.

\bibitem{DBLP:journals/bioinformatics/CardelliPTTVW21}
L.~Cardelli, I.~C. P{\'{e}}rez{-}Verona, M.~Tribastone, M.~Tschaikowski,
  A.~Vandin, and T.~Waizmann, ``Exact maximal reduction of stochastic reaction
  networks by species lumping,'' \emph{Bioinform.}, vol.~37, no.~15, pp.
  2175--2182, 2021.

\bibitem{DBLP:conf/birthday/CardelliTTV17}
L.~Cardelli, M.~Tribastone, M.~Tschaikowski, and A.~Vandin, ``Syntactic
  markovian bisimulation for chemical reaction networks,'' in \emph{Models,
  Algorithms, Logics and Tools - Essays Dedicated to Kim Guldstrand Larsen on
  the Occasion of His 60th Birthday}, L.~Aceto, G.~Bacci, G.~Bacci,
  A.~Ing{\'{o}}lfsd{\'{o}}ttir, A.~Legay, and R.~Mardare, Eds., vol. 10460,
  2017, pp. 466--483.

\bibitem{DBLP:conf/tacas/GhorbalP14}
K.~Ghorbal and A.~Platzer, ``Characterizing algebraic invariants by
  differential radical invariants,'' in \emph{{TACAS}}, E.~{\'{A}}brah{\'{a}}m
  and K.~Havelund, Eds., vol. 8413.\hskip 1em plus 0.5em minus 0.4em\relax
  Springer, 2014, pp. 279--294.

\bibitem{DBLP:conf/atva/BartocciKS19}
E.~Bartocci, L.~Kov{\'{a}}cs, and M.~Stankovic, ``Automatic generation of
  moment-based invariants for prob-solvable loops,'' in \emph{{ATVA}}, Y.~Chen,
  C.~Cheng, and J.~Esparza, Eds., 2019, pp. 255--276.

\bibitem{DBLP:journals/lmcs/Boreale19}
M.~Boreale, ``Algebra, coalgebra, and minimization in polynomial differential
  equations,'' \emph{Log. Methods Comput. Sci.}, vol.~15, no.~1, 2019.

\bibitem{DBLP:journals/scp/Boreale20}
------, ``Complete algorithms for algebraic strongest postconditions and
  weakest preconditions in polynomial odes,'' \emph{Sci. Comput. Program.},
  vol. 193, p. 102441, 2020.

\bibitem{DBLP:journals/tcs/TschaikowskiT14a}
M.~Tschaikowski and M.~Tribastone, ``Exact fluid lumpability in markovian
  process algebra,'' \emph{Theor. Comput. Sci.}, vol. 538, pp. 140--166, 2014.

\bibitem{DBLP:journals/tcs/TschaikowskiT14}
------, ``Tackling continuous state-space explosion in a markovian process
  algebra,'' \emph{Theor. Comput. Sci.}, vol. 517, pp. 1--33, 2014.

\bibitem{DBLP:journals/tcs/CardelliTTV19a}
L.~Cardelli, M.~Tribastone, M.~Tschaikowski, and A.~Vandin, ``Symbolic
  computation of differential equivalences,'' \emph{Theor. Comput. Sci.}, vol.
  777, pp. 132--154, 2019.

\bibitem{BBLTTV21Lics}
G.~Bacci, G.~Bacci, K.~G. Larsen, M.~Tribastone, M.~Tschaikowski, and
  A.~Vandin, ``Efficient local computation of differential bisimulations via
  coupling and up-to methods,'' in \emph{Symposium on Logic in Computer
  Science, {LICS}}, 2021, pp. 1--14.

\bibitem{DBLP:journals/tac/PappasLS00}
G.~J. Pappas, G.~Lafferriere, and S.~Sastry, ``Hierarchically consistent
  control systems,'' \emph{{IEEE} Trans. Automat. Contr.}, vol.~45, no.~6, pp.
  1144--1160, 2000.

\bibitem{DBLP:journals/tac/PappasS02}
G.~J. Pappas and S.~Simic, ``Consistent abstractions of affine control
  systems,'' \emph{{IEEE} Trans. Automat. Contr.}, vol.~47, no.~5, pp.
  745--756, 2002.

\bibitem{bisimulation_lin_sys_Schaft}
A.~J. van~der Schaft, ``Equivalence of dynamical systems by bisimulation,''
  \emph{IEEE Transactions on Automatic Control}, vol.~49, 2004.

\bibitem{DBLP:conf/cav/ClarkeGJLV00}
\BIBentryALTinterwordspacing
E.~M. Clarke, O.~Grumberg, S.~Jha, Y.~Lu, and H.~Veith, ``Counterexample-guided
  abstraction refinement,'' in \emph{Computer Aided Verification, 12th
  International Conference, {CAV} 2000, Chicago, IL, USA, July 15-19, 2000,
  Proceedings}, ser. Lecture Notes in Computer Science, E.~A. Emerson and A.~P.
  Sistla, Eds., vol. 1855.\hskip 1em plus 0.5em minus 0.4em\relax Springer,
  2000, pp. 154--169. [Online]. Available:
  \url{https://doi.org/10.1007/10722167\_15}
\BIBentrySTDinterwordspacing

\bibitem{DBLP:conf/cav/MoverCGIT21}
\BIBentryALTinterwordspacing
S.~Mover, A.~Cimatti, A.~Griggio, A.~Irfan, and S.~Tonetta, ``Implicit
  semi-algebraic abstraction for polynomial dynamical systems,'' in
  \emph{Computer Aided Verification - 33rd International Conference, {CAV}
  2021, Virtual Event, July 20-23, 2021, Proceedings, Part {I}}, ser. Lecture
  Notes in Computer Science, A.~Silva and K.~R.~M. Leino, Eds., vol.
  12759.\hskip 1em plus 0.5em minus 0.4em\relax Springer, 2021, pp. 529--551.
  [Online]. Available: \url{https://doi.org/10.1007/978-3-030-81685-8\_25}
\BIBentrySTDinterwordspacing

\bibitem{naldi2011dynamically}
A.~Naldi, E.~Remy, D.~Thieffry, and C.~Chaouiya, ``Dynamically consistent
  reduction of logical regulatory graphs,'' \emph{Theoretical Computer
  Science}, vol. 412, no.~21, pp. 2207--2218, 2011.

\bibitem{veliz2011reduction}
A.~Veliz-Cuba, ``Reduction of boolean network models,'' \emph{Journal of
  theoretical biology}, vol. 289, pp. 167--172, 2011.

\bibitem{AttractorsDefs}
J.~Milnor, ``On the concept of attractor,'' \emph{Communications in
  Mathematical Physics}, vol.~99, no.~2, pp. 177--195, 1985.

\bibitem{azpeitia2010single}
E.~Azpeitia, M.~Ben{\'\i}tez, I.~Vega, C.~Villarreal, and E.~R. Alvarez-Buylla,
  ``Single-cell and coupled grn models of cell patterning in the arabidopsis
  thaliana root stem cell niche,'' \emph{BMC systems biology}, vol.~4, no.~1,
  pp. 1--19, 2010.

\bibitem{RealsFragmentUndecidable}
D.~Richardson, ``Some undecidable problems involving elementary functions of a
  real variable,'' \emph{The Journal of Symbolic Logic}, vol.~33, no.~4, pp.
  514--520, 1968.

\bibitem{DBLP:journals/eatcs/Saxena09}
N.~Saxena, ``Progress on polynomial identity testing,'' \emph{Bull. {EATCS}},
  vol.~99, pp. 49--79, 2009.

\bibitem{Gear:1971}
C.~W. Gear, \emph{{Numerical Initial Value Problems in Ordinary Differential
  Equations}}.\hskip 1em plus 0.5em minus 0.4em\relax Upper Saddle River, NJ,
  USA: Prentice Hall PTR, 1971.

\bibitem{Hairer06}
E.~Hairer, C.~Lubich, and G.~Wanner, \emph{Geometric Numerical Integration},
  2006.

\bibitem{DBLP:journals/nc/CardelliTT20}
L.~Cardelli, M.~Tribastone, and M.~Tschaikowski, ``From electric circuits to
  chemical networks,'' \emph{Nat. Comput.}, vol.~19, no.~1, pp. 237--248, 2020.

\bibitem{AttractorsDiscretized}
P.~E. Kloeden and J.~Lorenz, ``Stable attracting sets in dynamical systems and
  in their one-step discretizations,'' \emph{SIAM Journal on Numerical
  Analysis}, vol.~23, no.~5, pp. 986--995, 1986.

\bibitem{naldi2009logical}
A.~Naldi, D.~Berenguier, A.~Faur{\'e}, F.~Lopez, D.~Thieffry, and C.~Chaouiya,
  ``Logical modelling of regulatory networks with ginsim 2.3,''
  \emph{Biosystems}, vol.~97, no.~2, pp. 134--139, 2009.

\bibitem{naldi2012efficient}
A.~Naldi, P.~T. Monteiro, and C.~Chaouiya, ``Efficient handling of large
  signalling-regulatory networks by focusing on their core control,'' in
  \emph{International Conference on Computational Methods in Systems
  Biology}.\hskip 1em plus 0.5em minus 0.4em\relax Springer, 2012, pp.
  288--306.

\bibitem{Drosophila}
A.~Faur{\'e}, B.~Vreede, E.~Sucena, and C.~Chaouiya, ``A discrete model of
  drosophila eggshell patterning reveals cell-autonomous and juxtacrine
  effects,'' \emph{PLoS Comput Biol}, vol.~10, p. e1003527, 2014.

\bibitem{cmsb2019tcs}
\BIBentryALTinterwordspacing
I.~C. Perez-Verona, M.~Tribastone, and A.~Vandin, ``A large-scale assessment of
  exact lumping of quantitative models in the biomodels repository,''
  \emph{Theoretical Computer Science}, 2021. [Online]. Available:
  \url{https://www.sciencedirect.com/science/article/pii/S0304397521003716}
\BIBentrySTDinterwordspacing

\bibitem{de2008z3}
L.~De~Moura and N.~Bj{\o}rner, ``Z3: An efficient smt solver,'' in
  \emph{International conference on Tools and Algorithms for the Construction
  and Analysis of Systems}.\hskip 1em plus 0.5em minus 0.4em\relax Springer,
  2008, pp. 337--340.

\bibitem{DBLP:conf/cmsb/ArgyrisLTTV22}
G.~Argyris, A.~Lluch{-}Lafuente, M.~Tribastone, M.~Tschaikowski, and A.~Vandin,
  ``An extension of {ERODE} to reduce boolean networks by backward boolean
  equivalence,'' in \emph{International Conference on Computational Methods in
  Systems Biology}, ser. LNCS, vol. 13447.\hskip 1em plus 0.5em minus
  0.4em\relax Springer, 2022, pp. 294--301.

\bibitem{chaouiya2013sbml}
C.~Chaouiya, D.~B{\'e}renguier, S.~M. Keating, A.~Naldi, M.~P. Van~Iersel,
  N.~Rodriguez, A.~Dr{\"a}ger, F.~B{\"u}chel, T.~Cokelaer, B.~Kowal
  \emph{et~al.}, ``{SBML} qualitative models: a model representation format and
  infrastructure to foster interactions between qualitative modelling
  formalisms and tools,'' \emph{BMC systems biology}, vol.~7, no.~1, pp. 1--15,
  2013.

\bibitem{naldi2015cooperative}
A.~Naldi, P.~T. Monteiro, C.~M{\"u}ssel, C.~for Logical~Models, Tools, H.~A.
  Kestler, D.~Thieffry, I.~Xenarios, J.~Saez-Rodriguez, T.~Helikar, and
  C.~Chaouiya, ``Cooperative development of logical modelling standards and
  tools with colomoto,'' \emph{Bioinformatics}, vol.~31, no.~7, pp. 1154--1159,
  2015.

\bibitem{dubrova2011sat}
E.~Dubrova and M.~Teslenko, ``A sat-based algorithm for finding attractors in
  synchronous boolean networks,'' \emph{IEEE/ACM transactions on computational
  biology and bioinformatics}, vol.~8, no.~5, pp. 1393--1399, 2011.

\bibitem{raza2008logic}
S.~Raza, K.~A. Robertson, P.~A. Lacaze, D.~Page, A.~J. Enright, P.~Ghazal, and
  T.~C. Freeman, ``A logic-based diagram of signalling pathways central to
  macrophage activation,'' \emph{BMC systems biology}, vol.~2, no.~1, pp.
  1--15, 2008.

\bibitem{DBLP:journals/tac/WhitbyCKLTT22}
M.~Whitby, L.~Cardelli, M.~Kwiatkowska, L.~Laurenti, M.~Tribastone, and
  M.~Tschaikowski, ``{PID} control of biochemical reaction networks,''
  \emph{{IEEE} Trans. Autom. Control.}, vol.~67, no.~2, pp. 1023--1030, 2022.

\bibitem{DBLP:journals/tcad/0001MTA16}
A.~Das, G.~V. Merrett, M.~Tribastone, and B.~M. Al{-}Hashimi, ``Workload change
  point detection for runtime thermal management of embedded systems,''
  \emph{{IEEE} Trans. Comput. Aided Des. Integr. Circuits Syst.}, vol.~35,
  no.~8, pp. 1358--1371, 2016.

\bibitem{DBLP:journals/corr/abs-2106-13898}
\BIBentryALTinterwordspacing
R.~M. Hasani, M.~Lechner, A.~Amini, L.~Liebenwein, M.~Tschaikowski, G.~Teschl,
  and D.~Rus, ``Closed-form continuous-depth models,'' \emph{CoRR}, vol.
  abs/2106.13898, 2021. [Online]. Available:
  \url{https://arxiv.org/abs/2106.13898}
\BIBentrySTDinterwordspacing

\end{thebibliography}
